\RequirePackage{fix-cm}
\documentclass{svjour3}                     
\smartqed  
\usepackage{graphicx}
 \usepackage{amssymb,amsfonts,amsmath}  
 \usepackage{xcolor}
\usepackage{url}
\usepackage{todonotes}
\newcommand{\R}{\mathbb R}
\newcommand{\pr}{\mathbf P}
\newcommand{\E}{\mathbb E}
\newcommand{\N}{\mathbb N}
\newcommand{\Z}{\mathbb Z}
\newcommand{\ind}{\mathbb{I}}
\begin{document}
\title{Detecting anomalies in fibre systems using 3-dimensional image data\thanks{This research was supported by the German Ministry of Education and Research (BMBF) through the project ``AniS'', the DFG Research Training Group GRK 1932, DFG grant SP 971/10-1, as well as the DAAD scientific exchange program ``Strategic Partnerships''.}
}

\author{Denis Dresvyanskiy \and Tatiana  Karaseva \and Vitalii Makogin \and Sergei Mitrofanov \and Claudia Redenbach \and Evgeny Spodarev
}

\institute{D. Dresvyanskiy \and   T.  Karaseva \and S. Mitrofanov \at
              Reshetnev Siberian State University of Science and Technology 
31, Krasnoyarsky Rabochy~Av., Krasnoyarsk, 660037, Russian Federation\\
              \email{DresvyanskiyDenis@yandex.ru, tatyanakarasewa@yandex.ru, sergeimitrofanov95@gmail.com} \and          
V. Makogin \and   E. Spodarev  \at
              Institut f\"{u}r Stochastik, Universit\"{a}t Ulm,
D-89069 Ulm \\
              \email{vitalii.makogin@uni-ulm.de}, \email{evgeny.spodarev@uni-ulm.de}   \and 
           C. Redenbach \at
Technische Universit\"{a}t Kaiserslautern,
Fachbereich Mathematik, Postfach 3049,\\ 67653 Kaiserslautern\\
\email{redenbach@mathematik.uni-kl.de} 
}

\maketitle

\begin{abstract}
We consider the problem of detecting anomalies in the directional distribution of  fibre materials observed in 3D images. We divide the image into a set of scanning windows and classify them into two clusters: homogeneous material and anomaly. Based on a sample of estimated local fibre directions, for each scanning window we compute several classification attributes, namely the coordinate wise means of local fibre directions, the entropy of the directional distribution, and a combination of them. 
We also propose a new spatial modification of the Stochastic Approximation Expectation-Maximization (SAEM) algorithm. Besides the clustering we also consider testing the significance of anomalies. To this end, we apply a change point technique for random fields and derive the exact inequalities for tail probabilities of a test statistics. The proposed methodology is first validated on simulated images. Finally, it is applied to a 3D image of a fibre reinforced polymer.

\end{abstract}

\section{Introduction}

{Fibre composites, e.g., fibre reinforced polymers or high performance concrete, are an important class of functional materials. Physical properties of a fibre composite such as elasticity or crack propagation are influenced by its microstructure characteristics including the fibre volume fraction, the size or the direction distribution of the fibres. Therefore, an understanding of the relations between the fibre geometry and macroscopic properties is crucial for the optimisation of materials for certain applications. During the last years, micro computed tomography (μCT) has proven to be a powerful tool for the analysis of the three-dimensional microstructure of materials.}

In the compression moulding process of glass fibre reinforced polymers, the fibres order themselves inside the raw material as a result of mechanical pressure. During this process, deviations from the requested direction may occur, creating undesirable fibre clusters and/or deformations. These inhomogeneities are characterized by abrupt changes in the direction of the fibres, and their detection is studied in this paper. 

The problem of detecting change points in random sequences, (multivariate) time series, panel and regression data has a long history, see the books \cite{BrassNik93,BrodDarkh00,Carlstein_ed94,ChenGupta12,TimeS1,Wu05}. Changes to be detected may concern the mean, variance, correlation, spectral density, etc. of the (stationary) sequence $\{X_k,k\geq 0\}$.  This kind of change detection has been considered by various authors starting with \cite{Page}. Sen and Srivastava \cite{Sen} considered tests for a change in the mean of a Gaussian model. An overview can also be found in \cite{Brodsky}. The CUSUM procedure, Bayesian approaches as well as maximum likelihood estimation are often used. Scan statistics come also into play naturally, see e.g. \cite{BrodDarkh99,BrodDarkh00}.

First approaches to change point analysis for random fields (or measures) have been developed in the papers \cite{Buc14,BucHeus15,CaoWors99,Chamb02,HanubiaMnatsakanov96,PiterJar11,Kaplan90,Kaplan92,Lai08,MuellSong94,Ninomiya04,Sharietal16,SiegYakir08,SiegWors95}, see also the review in \cite[Section 2, D]{BrodDarkh99} and \cite[Chapter 6]{BrodDarkh00}. The involved methods include M-estimation, minimax methods for risks, the geometric tube method, some nonparametric and Bayesian techniques.  However, much is still to be done in this relatively new area of research.

{In this paper, we develop a change-point test for $m-$dependent random fields. In the spirit of the book \cite{BrodDarkh00}, it uses inequalities for tail probabilities of suitable test statistics. It is applied to the mean and the entropy of the local directional distribution of fibres observed in a 3D image of a fibre composite obtained by micro computed tomography. Characteristics are estimated in a moving scanning window that runs over the observed material sample, cf. \cite{ruiz2016entropy,Alonzo}. Our main task is to detect areas with anomalous spreading of the fibres. Even though we focus on anomalies in fibres' directions, our method will work with any local characteristic of fibres with values in a (compact) Riemannian manifold such as fibre length or mean curvature.}

If an anomaly is present, its location is detected using a new spatial modification of the Stochastic Approximation Expectation Maximization (SAEM) algorithm (see \cite{EM} for a review of Expectation Maximization (EM) algorithms for the separation of components in a mixture of Gaussian distributions as well as a recent paper \cite{Laurent}). It allows for spatial clustering of the whole fibre material into a ``normal'' and an ``anomaly'' zones.

The paper is organized as follows. In Section 2,  we introduce the stochastic model of a fibre process. In Section 3, we describe the procedure of generating the sample data, introduce the mean of local directions as well as their entropy. There, we compare two methods for entropy estimation: plug-in and nearest neighbor statistics.  In Section 4, we consider the detection  of anomalies as a change-point problem for the corresponding $m-$dependent random fields. In Section 5, we localize the anomalous region of fibres solving a clustering problem for multivariate random fields. For this purpose, we propose a new spatial modification of SAEM algorithm, which decreases the diffuseness of clusters. In Section 6, we apply our methods to 3D images of simulated (Section 6.1) and real (Section 6.2) fibre materials and compare their performance.

\section{Problem setting}
In this section, we give some basic definitions and results for fibre processes. For more details, see, for example, the book \cite{stoyan}. In 3-dimensional Euclidean space, a fibre $\gamma$ is a simple curve  $\{\gamma(t)$ $=(\gamma_1(t),\gamma_2(t),\gamma_3(t)),$ $t\in [0,1]\}$ of finite length satisfying the following assumptions: 
\begin{itemize}
    \item $\{\gamma(t),t\in [0,1]\}$ is a $C^1$-smooth function.
    \item $\|\gamma'(t)\|_3^2>0$ for all $t\in[0,1],$ where $\|\gamma'(t)\|_3^2=|\gamma_1'(t)|^2+|\gamma_2'(t)|^2+|\gamma_3'(t)|^2.$
    \item A fibre does not intersect itself.
\end{itemize}

The collection of fibres forms a {\it fibre system} $\phi$ if it is a union of at most countably many fibres $\gamma^{(i)},$ such that any compact set is intersected by only a finite number of fibres, and $\gamma^{(i)}((0,1))\cap\gamma^{(j)}((0,1))=\varnothing,$ if $i\neq j,$ i.e., the distinct fibres may have only end-points in common. The length measure corresponding to the fibre system $\phi$ (and denoted by the same symbol) is defined by
$$\phi(B)=\sum_{\gamma^{(i)}\in \phi}h(\gamma^{(i)}\cap B) $$
 for bounded Borel sets  $B\in \mathcal{B}(\R^3),$ where $h(\gamma \cap B)=\int_0^1 \mathbb{I} \{\gamma(t)\in B\}\sqrt{|\gamma'(t)|^2}dt$ is the length of fibre $\gamma$ in window $B.$ Then $\phi(B)$ is the total length of fibre pieces in the window $B.$

\begin{definition}
A {\it fibre process} $\Phi$ is a random element with values in the set $\mathbb{D}$ of all fibre systems $\phi$ with $\sigma$-algebra $\mathcal{D}$ generated by sets of the form $\{\phi\in \mathbb{D}: \phi(B)<x\}$ for all bounded Borel sets $B$ and real numbers $x$.
The distribution $P$ of a fibre process is a probability measure on $[\mathbb{D},\mathcal{D}].$ The fibre process $\Phi$ is said to be {\it stationary} if it has the same distribution as the translated fibre process $\Phi_x=\Phi+x$ for all $x\in \R^3.$
\end{definition}

For classification needs we consider an abstract fibre characteristic $w$. Let $(E,\mathcal{E},\sigma)$ be a measurable space where $E$ is a (compact) Riemannian manifold equipped with a metric $\rho.$ Let $w(x)\in E$ be some characteristic of a fibre at point $x\in \R^3,$ assuming that exactly one fibre of $\Phi$ passes through $x.$
Then a weighted random measure $\Psi$ can be defined by
$$\Psi(B\times L)=\int_B \mathbb{I}\{w(x)\in L\}\Phi(dx)$$ for bounded $B\in\mathcal{B}(\R^3)$ and $L\in \mathcal{E}.$ Thus, $\Psi(B\times L)$ is the total length of all fibre pieces of $\Phi$  in $B$ such that their characteristic $w$ lies in range $L.$ 

As classifying characteristics $w$ we can for instance choose the fibres' local direction (with $E$ being the sphere $\mathbb{S}^2$), their length or curvature (both with $E=\mathbb{R}_+$). In this article we focus on local directions of fibres, but the results can easily be applied to other choices of $w.$ 

If the fibre process $\Phi$ is stationary then the {\it intensity measure} of $\Psi$  can be written as 
$\E \Psi(B\times L) =\lambda |B| f (L),$ where $\lambda>0$ is called the {\it intensity} of $\Psi,$ $|\cdot|$ is the Lebesgue measure in $\R^3$ and $f$ is a probability measure on $\mathbb{S}^2$ which is called the {\it directional distribution} of fibres. The distribution $f$ is the fibre direction distribution in the typical fibre point, hence length-weighted.  
In what follows, $|A|$ is either the cardinality of a finite set $A$ or the Lebesgue measure of $A$, if $A$ is uncountable and measurable.

Let $ \oplus$ and $ \ominus$ be the dilation (erosion, resp.) operation on images as introduced e.g. in \cite{stoyan}. Assume that we observe a dilated version $\Xi=\Phi \oplus B_r$ of $\Phi$ within a window $W=[a_1,b_1]\times[a_2,b_2]\times[a_3,b_3],$ $a_i<b_i, i=1,2,3,$ where $B_r$ is the ball of radius $r>0$ centered at the origin.  In our setting, we assume that the fibres' length is significantly larger than their diameter $2r.$ Moreover, we assume that there is $\varepsilon>0$ such that $\Xi$ is morphologically closed w.r.t. $B_\varepsilon,$ i.e.,  
$(\Xi\oplus B_\varepsilon)\ominus B_\varepsilon = \Xi.$ This condition ensures that the local fibre direction is uniquely defined in each point within $\Xi.$

We would like to test the hypothesis 
\begin{itemize}\label{hypothesis}
    \item[] $H_0:$ $\Phi$ is stationary with intensity $\lambda$ and directional distribution $f$ vs.
    \item[] $H_1:$ There exists a compact set $A \subset W$ with $|A|>0$ and $|W\setminus A|>0$ such that 
    \begin{align*}
        &\frac{1}{\lambda |A|  }\E \int_{A} \mathbb{I}\{w(x)\in \cdot\}\Phi(dx)\\
        &\neq \frac{1}{\lambda |W \setminus A| }\E \int_{W\setminus A} \mathbb{I}\{w(x)\in \cdot\}\Phi(dx).
    \end{align*}
\end{itemize}

If $H_1$ holds true, the region $A$ is called an {\it anomaly region.} In the following, we discuss how to test the hypothesis $H_0$ and how to detect the anomaly region $A$.

\begin{figure}[h]
\center{\includegraphics[width=0.45\linewidth,keepaspectratio]{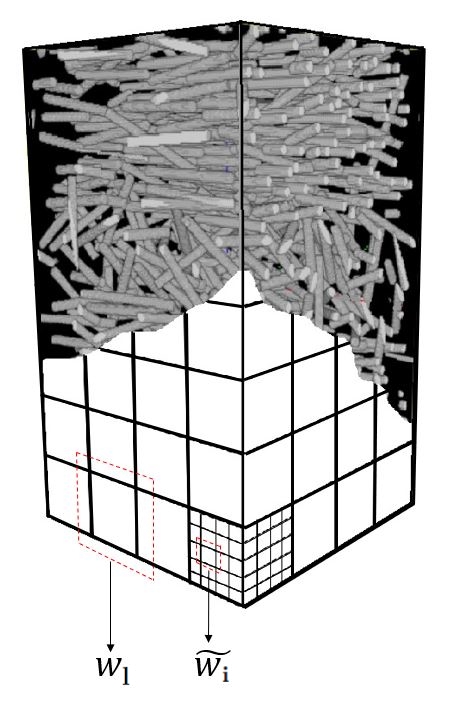}}
\caption{Construction of scanning windows $W_l.$}
	\label{fig:Picture1l}	
\end{figure}
\begin{figure}[h]
\center{\includegraphics[width=0.45\linewidth,keepaspectratio]{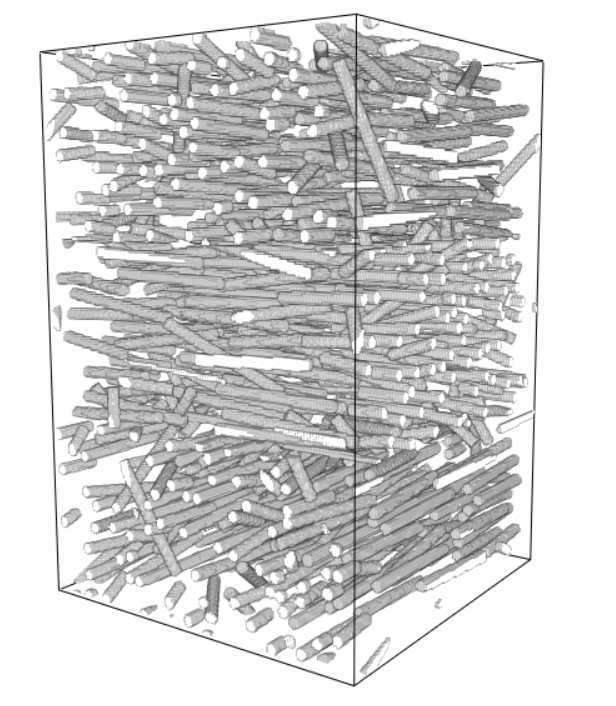}}
	\caption{Visualization of simulated material with anomaly zone.}
	\label{fig:Picture1r}
\end{figure}

\section{Data and clustering criteria}
We assume that the dilated fibre system $\Xi \cap W$ is observed as a 3D greyscale image. Several methods for estimating the local fibre direction $w(x)$ in each fibre pixel $x\in \Xi$ are discussed in \cite{ImageAnalStereol1489}. We use the approach based on the Hessian matrix that is implemented in the 3D image analysis software tool MAVI \cite{mavi}. The smoothing parameter $\sigma$ required by the method is chosen as $\sigma= \hat{r}$, where $\hat{r}$ is an estimate of the (constant) thickness of the typical fibre. In the simulated samples, it is known. For the real data, it is obtained manually from the images.

We divide the observation window $W$ into small cubes $\widetilde{W}_{\vec{i}}$ (see Fig.~\ref{fig:Picture1l}) of the same size, whose edge length $\Delta$ equals three times the fibre diameter. The principal axis $\hat{w}_{\vec{i}}$ of local directions (e.g. \cite{ImageAnalStereol1489}) in each $\widetilde{W}_{\vec{i}}$, here referred to as ``average local direction'', is then computed using the function \emph{SubfieldFibreDirections} in MAVI.

Let $J_T=\{\mathbf{i}=(i_1,i_2,i_3), i_1=\overline{1,n_1},i_2=\overline{1,n_2},i_3=\overline{1,n_3}\}$ be the regular grid of cubes $\widetilde{W}_{\vec{i}}.$ 
Some of the cubes $\widetilde{W}_{\vec{i}}$ may not contain enough fibre voxels to obtain a reliable estimate of the local fibre direction $\hat{w}_{\vec{i}}$. Let $J\subset J_T$ be the subset of indices of cubes which allow for such an estimation.
For each point $\vec{i}\in J$ denote by 
$X_{\vec{i}}=(x_{\vec{i}},y_{\vec{i}},z_{\vec{i}})^T$ the average local direction estimated from $\widetilde{W}_\mathbf{i}.$  We assume that our fibres are non-oriented and 
{can be then transformed such that}  $X_{\vec{i}}\in \mathbb{S}^2_+, $ where $\mathbb{S}^2_+ =\{(x,y,z)\in \R^3: x^2+y^2+z^2=1,z\in[0,1]\}$ is a hemisphere. The size of this sample is $N=|J|\leq  n_1 n_2 n_3.$

Our main task is to determine the anomaly regions or, in other words, to classify the set of points $J$ into two clusters corresponding to the \lq\lq{}homogeneous material\rq\rq{} and the \lq\lq{}anomaly\rq \rq{} (one of these clusters can be empty).
To do so, we combine  $M^3$ of the small  cubes $\widetilde{W}_{\vec{i}}$ (having edge length $\Delta$) to a larger cube $W_{\vec{l}}$, such that the 3D image $W$ is divided into larger non-empty cubic observation windows (see Figure \ref{fig:Picture1l}). The larger cubes have side length $M\Delta$ and the corresponding grid of the larger cubes is denoted by $J_W=\{\vec{l}=(l_1,l_2,l_3):l_1=\overline{1,m_1}, l_2=\overline{1,m_2}, l_3=\overline{1,m_3}\}$.
The set of indexes of non-empty cubes $\widetilde{W}_{\vec{i}}$ within $W_{\vec{l}}$ is denoted by $S_{\vec{l}}=\{(i_1,i_2,i_3)\in J,\widetilde{W}_{(i_1,i_2,i_3)}\subset W_{\vec{l}}\}$.
For each window $W_{\vec{l}},\vec{l}\in J_W$, we estimate the entropy and the mean of the local directions, based on the estimates $\{\hat{w}_{\vec{i}},\vec{i}\in S_{\vec{l}} \}$ as described below.

\subsection{Mean of local fibres direction}
The vector $M_{\vec{l}}=(M_{x,\vec{l}},M_{y,\vec{l}},M_{z,\vec{l}})^T$ is calculated for each window $W_{\vec{l}},\vec{l}\in J_W$ as the coordinate-wise sample mean of local directions (MLD):
\begin{equation*}\begin{gathered}M_{x,\vec{l}}=\frac{1}{|S_{\vec{l}}|}\sum_{\vec{i}\in S_{\vec{l}}} \hat{w}_{i_1},\quad M_{y,\vec{l}}=\frac{1}{|S_{\vec{l}}|}\sum_{\vec{i}\in S_{\vec{l}}} \hat{w}_{i_2},\\
~ M_{z,l}=\frac{1}{|S_{\vec{l}}|}\sum_{\vec{i}\in S_{\vec{l}}} \hat{w}_{i_3} .
\end{gathered}
\end{equation*}
Note that $|S_{\vec{l}}|\leq M^3 $ and normally $|S_{\vec{l}}|\approx M^3.$

\subsection{The entropy of the directional distribution}

The entropy of a random variable is a certain measure of the diversity/concentration of its range. Let $\mathbf{P}$ be a probability distribution of a random element $X$ on an abstract measurable phase-space $(\mathbb{X},\sigma).$ The value
\begin{equation}
\label{entropy1}
    E_X = - \int_{\Omega} \ln( \mathbf{P}(X(\omega))) \mathbf{P}(d\omega)
\end{equation}
is called the {\it Shannon (differential) entropy} of $X$.  If  $\mathbf{P}$ is absolutely continuous with respect to some  measure $\sigma$ then there exists the Radon-Nikodym derivative (or density) $f=\frac{d \mathbf{P}}{d \sigma},$ and the entropy of $X$ has the following form
\begin{equation}
\label{entropy2}
    E_f:= E_X = - \int_{\mathbb{X}} \ln( f(x)) f(x)\sigma(d x).
\end{equation}
In what follows, we assume that the random variable $X$ is absolutely continuous with density $f:\mathbb{X}\to \R_+.$ In our problem setting, $X$ corresponds to the local fibre direction, $\mathbb{X}$ is the sphere $\mathbb{S}^2$ and $\sigma$ is the spherical surface area measure on $\mathbb{S}^2.$  Since our fibres are non-oriented ($X\in \mathbb{S}^2_+$), we may consider even local direction densities $f$ on the whole sphere $\mathbb{S}^2$ where appropriate, instead of a density $f$ defined on $\mathbb{S}^2_+.$
However, choosing another classifying characteristic $w$ will lead to a different measurable space $(\mathbb{X},\sigma).$  

\subsection{Entropy estimation}

In the literature, there is a large number of papers devoted to non-parametric entropy estimation for i.i.d. random vectors in $\R^D,$ see e.g. the review in \cite{BDGM} and references in \cite{BD}. We will dwell upon two important estimates: the plug-in and the nearest neighbor ones. 

\subsubsection {Plug-in estimator of entropy}
For simplicity, define the plug-in estimator for directional distributions on the sphere $\mathbb{S}^2$ with even densities $f:$ $f(x)=f(-x),x\in \mathbb{S}^2.$ Its general definition on compact manifolds $\mathbb{X}$ can be found e.g. in \cite{Alonzo}. 

For a directional distribution density $f,$ take the kernel estimator $\widehat{f}_B(\cdot)$
 on a window $B\subseteq J_T$ of the form

\begin{equation}
\label{kernel}
  \widehat{f}_{B}(y) =\frac{1}{|B|}\sum \limits_{\mathbf{i} \in B\cap J }\frac{|\sin \rho(y,X_\mathbf{i})|}{h^2\rho(y,X_\mathbf{i})} K\left(\frac{\rho(y,X_\mathbf{i})}{h}\right),
\end{equation}
where $h > 0$ denotes the bandwidth, $K : \R_+ \to \R$ is a kernel function and $\rho:\mathbb{S}^2\times \mathbb{S}^2 \to \R_+$ is a geodesic metric  given by $\rho(\mathbf{x},\mathbf{y})=\arccos \langle \mathbf{x},\mathbf{y} \rangle, \mathbf{x},\mathbf{y}\in \mathbb{S}^2,$ where  $\langle \cdot,\cdot \rangle$ is the Euclidean scalar product in $\R^3$.

Then the {\it plug-in} estimator of $E_{f}$ in the window $W_\mathbf{l}$ is given by
\begin{equation}
\label{plugin}
 \widehat{E}_{f, \mathbf{l}}=-\frac{1}{|S_\mathbf{l}|} \sum_{\mathbf{i}\in S_\mathbf{l}}\ln \widehat{f}_{B+\mathbf{i}}(X_\mathbf{i}),
\end{equation}
 where $B \subset J_T$ is the sub-window and $B+\mathbf{i}:=\{\mathbf{j}+\mathbf{i},\mathbf{j}\in B\}$ denotes the translation of $B.$ 

For homogeneous marked Poisson point processes, the plug-in estimator $\widehat{E}_{f}$ as above is considered in \cite{Alonzo}. See also \cite{ruiz2016entropy} for the context of Boolean models of line segments. We also made an attempt to apply this method to our 3D image data. But we met  difficulties which basically come from the relatively small amount of data available. 
Namely, $\widehat{E}_{f}$ needs a large number of points in sub-windows $B'$ during the estimation of $f$ together with a large number of such sub-windows. Let us illustrate these difficulties on a simple example.

\begin{example}
Consider the uniform distribution on the sphere $\mathbb{S}^2,$ i.e., the density is $f(x)=\frac{1}{4\pi},x\in \mathbb{S}^2.$ We generate a sample from this distribution and estimate its entropy $E_f$ using the plug-in estimator \eqref{plugin} with $|B|=|\mathbb{S}^2|^{1/9}$ (as in \cite{ruiz2016entropy}). The results are presented in Table \ref{tabular:table1}.
 Moreover, we run the procedure 100 times and compare the obtained values with the exact value of the entropy 
\[ E_f = 
-\int\limits_{\mathbb{S}^2}\dfrac{1}{4\pi} \ln \dfrac{1}{4\pi} \sigma(dx)=\ln 4\pi \approx 2.53. \]

Obviously, the bias of $\widehat{E}_{f}$ is too large with less than 62000 entries, which is in accordance with \cite{M} stating the impracticability of plug-in entropy estimates for samples in higher dimensions.  There are 430741 entries in $J_T$ for the real data (Figure \ref{fig:Picture22l}) and 463537 entries in RSA fibre data (Figure \ref{fig:Picture5l}), that allows us to subdivide the images only into 4 non-intersecting regions with more than 100000 cubes $W_\mathbf{l}.$ In other words, in order to test the hypotheses $H_0$ vs. $H_1$ with test statistics based on estimated entropy \eqref{plugin} we have a sample of size 4, which is too small, compare Section 4, inequality \eqref{msize}. There, for $m=1$ the minimal sample size $|W|$ must be 1000.

\end{example}
\begin{table}
	\caption{The plug-in entropy estimator for the uniform directional distribution on $\mathbb{S}^2$}
	\label{tabular:table1}
		\begin{tabular}{rrcc}
			\hline\noalign{\smallskip}
			\textbf{Sample size}	& \textbf{Mean}	& \textbf{Variance}& \textbf{MSE} \\
			\noalign{\smallskip}\hline\noalign{\smallskip}
			500000 & 2.476  & $2.848\times10^{-6}$ & 0.003\\
			375000	& 2.456	& $5.034\times 10^{-6}$ & 0.006\\
			250000	& 2.418	& $6.293\times10^{-6}$ & 0.013\\
			125000	& 2.309	& $7.245\times10^{-6}$ & 0.049\\
			62500	& 2.099	& $2.739\times10^{-5}$ & 0.187\\
			12500	& 0.981	& $8.917\times10^{-5}$ & 2.403\\
			6250	& 0.359	& $1.251\times10^{-4}$ & 4.718\\
			1250	& 0.008	& $4.695\times10^{-4}$  & 6.36\\
			\noalign{\smallskip}\hline
		\end{tabular}
\end{table}

\subsubsection {Nearest neighbor estimator of entropy}
In order to overcome the above difficulties we apply another estimator of $E_f$ introduced in the paper by Kozachenko and Leonenko \cite{Leon}. We call this estimator ``Dobrushin estimator'' because its main idea is due to Dobrushin \cite{Dobr}. The estimator from \cite{Leon} cannot be applied directly, because it is designed for random vectors in a $d$-dimensional Euclidean space which is flat. In our setting, the phase space  $(\mathbb{S}^2,\sigma)$ is a manifold of positive constant curvature with geodesic metric $\rho.$ Therefore, we take a version of Dobrushin estimator for the case of an $d-$dimensional compact Riemannian manifold $\mathbb{X}$ with geodesic metric $\rho$ and Hausdorff measure $\sigma$. 

For defining this estimator, the following results will be useful.
Denote by $B_\delta(x)$ the ball in $(\mathbb{X},\rho)$ with radius $\delta>0$ and center $x,$ i.e., $B_\delta(x)=\{y\in \mathbb{X}: \rho(x,y)\leq \delta\}.$  Since a $\rho-$ball and a Euclidean $d$-dimensional ball are bi-Lipschitz equivalent, $d$ coincides with the Hausdorff dimension of the manifold $\mathbb{X},$ see \cite[Corollary 2.4]{Fal}.  Furthermore, for $\sigma-$almost all points $x\in \mathbb{X}$ the Lebesgue density theorem is true, i.e., \begin{equation}
\label{sigma:a}
\sigma(B_\delta(x))\sim c \delta^d,
\, \text{ as }\delta\to0+, 
\end{equation}
where $d=\dim_H X$ is the Hausdorff dimension of $\mathbb{X}$ and $c=2^d D_X>0,$ where $D_X$ is the Hausdorff density of $\mathbb{X}$, see \cite[Proposition 4.1,5.1]{Fal} and \cite[Theorem 30]{rogers}.

\begin{definition}
Let $(\xi_1,\ldots,\xi_N)$ be a sample of i.i.d. $\mathbb{X}-$ valued random elements with continuous density function $f:\mathbb{X}\to \R_+.$ Denote by $\rho_i$ the distance to the nearest neighbor of $\xi_i, i=\overline{1,N},$ i.e., $\rho_i=\min_{j=\overline{1,N},j\neq i} \{\rho(\xi_i,\xi_j)\}.$ Define the statistic
\begin{equation}
\label{dobr_est}
\widehat{E}=\frac{d}{N}\sum_{i=1}^N\ln \rho_i+\ln (c(N-1))+ \gamma,
\end{equation}
where $c$ and $d$ are defined by \eqref{sigma:a} and $\gamma=-\int_{0}^{\infty}(\ln y)e^{-y}dy\approx 0.5772$ is Euler's constant. The statistic $\widehat{E}$ is called {\it nearest neighbor (Dobrushin) estimator} of the entropy. 
\end{definition} 
It coincides with the nearest-neighbour entropy estimate given in \cite[p. 2169]{Penrose} with the only difference that in \cite{Penrose} Euclidean distances between $\xi_i$ are used instead of geodesic distances $\rho_i.$ The $L_2$-consistency of $\widehat{E}$ is proven in \cite[Theorem 2.4]{Penrose} for i.i.d. samples $(\xi_1,\ldots,\xi_N)$ as above with bounded density $f$ of compact support.

In fact, a large class of parametric distributions on a sphere, including the Fisher, the Watson or the Angular Central Gaussian distribution, has bounded densities with compact support.
\begin{remark}
In many problems of probability theory, limit theorems for independent observations remain true for weakly dependent data. Since the fibre materials are weakly dependent (the fibers have a finite length), we can assume that the entropy and mean local directions are weakly dependent as well. The proof of consistency of \eqref{dobr_est} for weakly dependent $\xi_i$ is non-trivial and goes beyond the scope of this paper.
\end{remark}
\begin{remark}
Our data sets consist of straight fibres which are longer than the edge length $\Delta$ of small observation windows $\widetilde{W}.$ Such fibres yield several almost equal values of average local directions $X_\mathbf{i}.$ This leads to very small values of a distance to the nearest neighbor $\rho_i$ and, consequently, to the large negative bias of $\hat{E}$ which is computed using $\log \rho_i.$ Trying to eliminate this bias, we propose to use the following version of \eqref{dobr_est} 
\begin{equation}\begin{gathered}
    \label{dobr_est_v}
    \widehat{E}=\frac{d}{\sum_{i=1}^N \ind\{\rho_i>\rho_0\}}\sum_{i=1}^N \ind\{\rho_i>\rho_0\}\ln \rho_i+\ln c \\
    +\ln \left(\sum_{i=1}^N \ind\{\rho_i>\rho_0\}-1\right)+ \gamma
    \end{gathered}
\end{equation}
with  penalty value $\rho_0=0.01$ found by computational tuning.
\end{remark}

In order to test the accuracy of the Dobrushin estimator, we have generated 100 samples from the uniform directional distribution on $\mathbb{S}^2$. We have computed the Dobrushin statistic and compared it with the exact entropy value $\ln (4\pi)\approx 2.53.$ The results are presented in Table \ref{tabular:table01}. 
\begin{table}
	\caption{The Dobrushin estimator \eqref{dobr_est} for the uniform directional distribution on the sphere.}
	\label{tabular:table01}
	\begin{center}
		\begin{tabular}{rrccc}
		\hline\noalign{\smallskip}
			\textbf{Sample size}	& \textbf{Mean}	&  \textbf{Variance} &\textbf{MSE}\\
			\noalign{\smallskip}\hline\noalign{\smallskip}
			125 & 2.51 & 0.02 &  0.02\\
			64	& 2.50	& 0.03 &   0.02\\
			\noalign{\smallskip}\hline
		\end{tabular}
	\end{center}
\end{table}

Based on these results, we conclude that the Dobrushin estimator \eqref{dobr_est} is quite accurate for small sample sizes. Even for a sample with 64 entries the entropy is estimated much better than by the plug-in method.

\section{Change point detection in random fields}
To test the hypothesis $H_0$ against $H_1,$ we check the existence of anomaly regions in a realization of an $m-$dependent geometric random field $\{s_k,k\in W\}.$ 
Here we follow the ideas from \cite{Brodsky}, where change-point problems for mixing random fields on general parametric (disorder) regions were considered. The field 
$\{s_k,k\in W\}$ will be chosen in a way such that the hypothesis $H_0$ implies that it is stationary, whereas $H_1$ means the presence of a region $I_\theta\subset W$ with different mean value of $s_k.$
Later in our application to fibre materials in Section 4.3, we assume the anomaly region to be a box \cite{BucchiaWendler17}.

\subsection{Random fields with inhomogeneities in mean}
Let $\{\tilde{\xi}_k,k\in \Z^3\}$ be an integrable stationary real-valued random field with $\mu=\E \tilde{\xi}_k.$ Denote by $\xi_k$ the  centered field $\xi_k=\tilde{\xi}_k-\mu, k=(k_1,k_2,k_3)\in \Z^3.$ 
Moreover, we assume that $\{\xi_k,k\in \Z^3\}$ is {\it $m-$dependent}, i.e, $\xi_k$ and $\xi_l$ are independent if $\max_{i=1,2,3}|k_i-l_i|>m, k=(k_1,k_2,k_3)\in \Z^3, l=(l_1,l_2,l_3)\in \Z^3.$
Let $\Theta$ be a finite parameter space. For every $\theta\in \Theta$ we define subsets of anomalies $I_\theta\subset \Z^3$  completely determined by a parameter $\theta.$
Then for some $\theta_0\in \Theta$ we observe \begin{equation}
\label{xk:def}s_k=\tilde{\xi}_k+h \ind\{k\in I_{\theta_0}\},~k\in W,
\end{equation}
where $W=[1,M_1]\times [1,M_2] \times [1,M_3] \cap \Z^3,$ and $h\in \R.$ Assume that $I_\theta\subset W$ for every $\theta\in \Theta.$ Denote $I_\theta^c=W\setminus I_\theta.$

Let $\Theta_0$ correspond to the values of $\theta$ for anomalies which we consider as significant, i.e., they are neither extremely small nor represent the majority of the data.
Formally, for $\gamma_0,\gamma_1\in (0,1),$ $\gamma_0<\gamma_1,$ we let 
\begin{equation*}
    \Theta_0=\{\theta\in \Theta: \gamma_0|W|\leq |I_\theta|\leq\gamma_1|W|\}.
\end{equation*}
Then $\Theta_1= \Theta \setminus \Theta_0$ corresponds to extremely small or large anomalies, i.e., 
\begin{equation*}
    \Theta_1=\{\theta\in \Theta: |I_\theta|<\gamma_0|W|, \text{ or }|I_\theta^c|<(1-\gamma_1)|W|\}.
\end{equation*}

\subsection{Testing the change of expectation}
Now we can formulate the change-point hypotheses for the random field $\{s_k,k\in W\}$ with respect to its expectation as follows.
\begin{itemize}
     \item[] $H_0':$  $\E s_k=\mu$ for every $k\in W$ (i.e.  $h=0$) vs.
     \item[] $H_1':$  There exists $\theta_0\in \Theta_0$ such that $\E s_k=\mu+h, k\in I_{\theta_0}, ~ h\neq 0,$ and $\E s_k=\mu,k\in I_{\theta^c_0}.$
 \end{itemize}
Consider the following change-in-mean statistics for the sample $S=\{s_k,k\in W\}:$
\begin{align}
    &Z(\theta)=\frac{1}{|I_\theta|}\sum_{k\in I_\theta} s_k-\frac{1}{|I^c_\theta |}\sum_{k\in I^c_\theta} s_k\\
 \nonumber   &=\frac{1}{|I_\theta |}\sum_{k\in I_\theta} (\xi_k + \mu +h \ind\{k\in I_{\theta_0}\})\\
 \nonumber &-\frac{1}{|I^c_{\theta_0}|}\sum_{k\in I^c_\theta } (\xi_k + \mu +h \ind\{k\in I_{\theta_0}\})\\
    &=\frac{1}{|I_\theta |}\sum_{k\in I_\theta } \xi_k -\frac{1}{|I^c_\theta|}\sum_{k\in I^c_\theta } \xi_k +h\left(\frac{|I_\theta \cap I_{\theta_0}|}{|I_\theta |}-\frac{|I^c_\theta \cap I_{\theta_0}|}{|I^c_\theta|}\right).
\end{align}
Denote by 
\begin{equation*}
    \begin{gathered}\eta(\theta):=Z(\theta)-\E Z(\theta)\\
    = \frac{1}{|I_\theta|}\sum_{k\in I_\theta } \xi_k -\frac{1}{|I^c_\theta|}\sum_{k\in I^c_\theta } \xi_k, ~\theta \in \Theta
    \end{gathered}
\end{equation*}
the centered field $Z(\theta).$
In order to test $H'_0$ vs. $H'_1$ we use the test statistics 
\begin{equation}
\label{teststat}
\begin{gathered}
T_W(S)=\max_{\theta\in \Theta_0}|Z(\theta)|\\
=\max_{\theta\in \Theta_0}\left|\frac{1}{|I_\theta|}\sum_{k\in I_\theta} s_k-\frac{1}{|I^c_\theta |}\sum_{k\in I^c_\theta} s_k\right|.
\end{gathered}
\end{equation}
We reject $H_0'$ if $T_W(S)$ exceeds the critical value $y_\alpha.$ Let us find such $y_\alpha>0$ via the probability of the 1st-type error $\pr_{H_0'}(\max_{\theta\in \Theta_0}|Z(\theta)|\geq y_\alpha)\leq \alpha.$ 
It holds
\begin{align*}
    &\pr_{H_0'}(\max_{\theta\in \Theta_0}|Z(\theta)|\geq y_\alpha )\\
    &=\pr_{H_0'}\left(\max_{\theta\in \Theta_0}\left|\eta(\theta)+h\left(\frac{|I_\theta \cap I_{\theta_0}|}{|I_\theta |}-\frac{|I^c_\theta \cap I_{\theta_0}|}{|I^c_\theta|}\right)\right|\geq y_\alpha\right)\\
    &=\pr(|\max_{\theta\in \Theta_0}|\eta(\theta)|\geq y_\alpha).
\end{align*}
Thus, we find the bounds for tail probabilities of the  random variable $\max_{\theta\in \Theta_0}|\eta(\theta)|.$ To do so, we use the ideas from \cite{Heinrich} to get the following bounds for $m-$dependent random fields. For the sake of generality, our results are formulated in $\mathbb{Z}^D, D\in \N.$
\begin{theorem}
\label{thm2}
 Let $\{\xi_k,k\in \Z^D\}$ be a stationary real-valued $m-$dependent centered random field and $\{b_k,k\in \Z^D\}$ be real numbers. Assume that there exist $H,\sigma>0$ such that 
 \begin{equation}
 \label{thm:c}       \E |\xi_k^p| \leq \frac{p!}{2}H^{p-2} \sigma^2,~ p=2,3,\ldots
 \end{equation}
Then for any $W\subset \Z^D,$ $|W|<\infty$ we have
\begin{align*}
    &\pr\left(\left|\sum_{k\in W} \xi_k b_k\right|\geq y\right)\leq 
    2\exp\left(- \frac{y^2}{4 m^D \sigma^2 \|b\|_2^2}\right)\\
    &\times \ind\left\{0<y\leq \frac{\sigma^2\|b\|_2^2}{H \|b\|_{\infty}}\right\}\\
   & +
    2 \exp\left(- \frac{y}{2 H m^D \|b\|_{\infty}}  + \frac{\sigma^2\|b\|^2_{2}}{4 H^2 m^D  \|b\|^2_{\infty}} \right)\\
    &\times\ind\left\{y> \frac{\sigma^2\|b\|_2^2}{H \|b\|_{\infty}}\right\},
\end{align*}
where $\|b\|_{\infty}=\max_{k\in W}|b_k|$ and
$\|b\|_2^2=\sum_{k\in W}b^2_k.$
\end{theorem}
\begin{proof}
Using the Markov inequality we have for any $u>0$ that
\begin{align}
\nonumber& \pr\left(\left|\sum_{k\in W} \xi_k b_k\right|\geq y\right)=\pr\left(\sum_{k\in W} \xi_k b_k\geq y\right)\\
\nonumber&+\pr\left(\sum_{k\in W} \xi_k (-b_k)\geq y\right) \leq  \frac{\E \exp( u\sum_{k\in W} \xi_k b_k)}{\exp(uy)}\\
\label{Heq1}   &+\frac{\E \exp( u\sum_{k\in W} \xi_k (-b_k))}{\exp(uy)}.
\end{align}
Denote by $W(l)=\{l+m i\in W| i\in \Z^D\}$ for $l\in \{1,\ldots,m\}^D.$ It follows from H\"{o}lder's inequality and $m-$dependence that
\begin{align}
\nonumber   &\E \exp\left( u\sum_{k\in W} \xi_k b_k\right)\\
\nonumber   &= \E \left(\prod_{l\in\{1,\ldots,m\}^D}\exp\left( u\sum_{k\in W(l)} \xi_k b_k\right)\right)\\
\nonumber   &\leq \prod_{l\in\{1,\ldots,m\}^D} \left(\E\exp\left( m^D u\sum_{k\in W(l)} \xi_k b_k\right)\right)^{1/m^D}\\
\nonumber&=\prod_{l\in\{1,\ldots,m\}^D} \left(\prod_{k\in W(l)}\E e^{m^D u \xi_k b_k}\right)^{1/m^D}\\
\label{Heq2}   &=\prod_{k\in W} \left(\E e^{m^D u  b_k \xi_k}\right)^{1/m^D}.
\end{align}
From Taylor's expansion we have for $u \in \left[0,\frac{1}{2 H m^D  |b_k|)^{-1}}\right]$
\begin{align}
\nonumber  &\E e^{m^D u  b_k \xi_k} =1+ m^D u  b_k \E \xi_k + \frac{1}{2} m^{2d} u^2  b^2_k\E \xi_k^2\\
\nonumber&+ m^{2d} u^2 b^2_k \sum_{p\geq 3}\frac{1}{p!} m^{d (p-2)} u^{p-2}  b^{p-2}_k\E \xi_k^p\\
\nonumber  &  \leq 1+  \frac{1}{2} m^{2d} u^2  b^2_k\sigma^2 + \frac{1}{2}m^{2d} u^2 b^2_k \sigma^2 \sum_{p\geq 3}(H m^D u |b_k|)^{p-2}\\
\nonumber&=1+ \frac{1}{2}\frac{m^{2d} u^2  b^2_k\sigma^2}{1-H m^D u |b_k|}\\
\label{Heq3}  &\leq 1+ m^{2d} u^2  b^2_k\sigma^2 \leq \exp(m^{2d} u^2  b^2_k\sigma^2).
\end{align}
Combining  \eqref{Heq2} and \eqref{Heq3} we continue for the first term in \eqref{Heq1} with the following bound for
$0 \leq u \leq $ \linebreak $ (2 H m^D  \max_{k\in W}|b_k|)^{-1}:$
\begin{equation}
\label{Heq4}
\begin{gathered}
      e^{-uy}\prod_{k\in W} \left(\E e^{m^D u  b_k \xi_k}\right)^{1/m^D} \\
      \leq  
    \exp\left(- uy +m^D u^2 \sigma^2\sum_{k\in W}b^2_k \right).
\end{gathered}
\end{equation}
The minimum of \eqref{Heq4} is achieved for $u=y/(2 m^D \sigma^2 \|b\|_2^2).$ Moreover,  bound \eqref{Heq4} is valid for the second term in \eqref{Heq1}, too. Therefore, for $0\leq y\leq \frac{\sigma^2\|b\|_2^2}{H \|b\|_{\infty}}$ we have
$$\pr\left(\left|\sum_{k\in W} \xi_k b_k\right|\geq y\right)\leq 2 \exp\left(- \frac{y^2}{4 m^D \sigma^2 \|b\|_2^2}\right). $$
For $y> \frac{\sigma^2\|b\|_2^2}{H \|b\|_{\infty}}$ we put $u=(2 H m^D \|b\|_{\infty})^{-1}$ in \eqref{Heq4} and obtain 
\begin{align*}
    &\pr\left(\left|\sum_{k\in W} \xi_k b_k\right|\geq y\right)\\
    &\leq 2 \exp\left(- \frac{y}{2 H m^D \|b\|_{\infty}}  + \frac{1}{4 H^2 m^{D}  \|b\|^2_{\infty}} \sigma^2\|b\|^2_{2} \right).
\end{align*}
This completes the proof.
\end{proof}
We apply Theorem \ref{thm2} to  $\eta(\theta),\theta\in \Theta_0,D=3.$ First, we rewrite $\eta(\theta)$ as 
\begin{equation*}
    \begin{gathered}
    \eta(\theta)= \frac{1}{|I_\theta|}\sum_{k\in W} \xi_k \ind\{k\in I_\theta\} -\frac{1}{|I^c_\theta|}\sum_{k\in W} \xi_k \ind\{k\in I_\theta^c\}\\
    =\sum_{k\in W} \xi_k b_k(\theta), \theta \in \Theta_0,
    \end{gathered}
\end{equation*}
where
\begin{equation*}
b_k(\theta)=  \frac{\ind\{k\in I_\theta\}}{|I_\theta|} -\frac{\ind\{k\in I_\theta^c\}}{|I^c_\theta|}.
\end{equation*}
\begin{corollary}
\label{cor2}
Let $|I_\theta|\leq |I_\theta^c |$ for $\theta\in \Theta_0$ and $|\xi_k|\leq M_0,k\in W$ a.s., then under the conditions of Theorem \ref{thm2} we have that
\begin{align*}
    &\pr\left(\left|\eta(\theta)\right|\geq y\right)\\
    &\leq 
    2\exp\left(- \frac{y^2}{4 m^3 \sigma^2 }\frac{|I^c_\theta||I_\theta|}{|W|}\right)\ind\left\{0<y\leq \frac{\sigma^2|W|}{M_0 |I_\theta^c|}\right\}\\
   & +
    2\exp\left(- \frac{y}{2 M_0 m^3}|I_\theta|+\frac{\sigma^2 |W|}{4 M_0^2 m^3 |I_\theta^c|}|I_\theta|\right)\\
    &\times\ind\left\{y> \frac{\sigma^2|W|}{M_0 |I_\theta^c  |}\right\}.
\end{align*}
\end{corollary}
\begin{proof}
From the definition of $b_k(\theta)$ we have
\begin{align*}
&\|b(\theta)\|^2_{2}=\sum_{k\in W}\left(\frac{\ind\{k\in I_\theta \}}{|I_\theta |} -\frac{\ind\{k\in I_\theta^c \}}{|I^c_\theta|}\right)^2\\
&=\frac{1}{|I_\theta|}+\frac{1}{|I^c_\theta|}=\frac{|W|}{|I^c_\theta||I_\theta|},
\end{align*}
and
\begin{align*}
    \|b(\theta)\|_{\infty}&=\max_{k\in W}|b_k(\theta)|=\max\left(\frac{1}{|I_\theta|},\frac{1}{|I^c_\theta |}\right)=\frac{1}{|I_\theta |}.
\end{align*}
Since $|\xi_k|\leq M_0$ then $\E |\xi_k^p| \leq \E (\xi^2_k |\xi_k|^{p-2}) \leq M_0^{p-2} \E \xi^2_k$ $\leq M_0^{p-2} \sigma^2.$
Therefore,  we put $H=M_0$ in \eqref{thm:c} and obtain
the statement of the corollary.
\end{proof}
Since $\pr\left(\max_{\theta\in \Theta_0}|\eta(\theta)|\geq y\right)\leq \sum_{\theta\in \Theta_0}\pr\left(|\eta(\theta)|\geq y\right),$ 
we have the following corollary.
\begin{corollary}
\label{cor3}
Let $0<\gamma_0<\gamma_1<1/2$ and assume that the conditions of Theorem \ref{thm2} and Corollary \ref{cor2} are satisfied. Then 
\begin{align}
\nonumber    &\pr\left(\max_{\theta\in \Theta_0}|\eta(\theta)|\geq y\right)\\
\label{cor3eq1}&\leq \sum_{\theta\in \Theta_0: |I_\theta^c|\leq \frac{\sigma^2 |W|}{y M_0}}2e^{- \frac{y^2}{4 m^3 \sigma^2 }\frac{|I^c_\theta||I_\theta|}{|W|}}\\
\label{cor3eq2}    &+\sum_{\theta\in \Theta_0: |I_\theta^c|> \frac{\sigma^2 |W|}{y M_0}}2e^{- \frac{y}{2 M_0 m^3}|I_\theta|+\frac{\sigma^2 |W|}{4 M_0^2 m^3 |I_\theta^c|}|I_\theta|}.
\end{align}
\end{corollary}
Hence, we reject $H_0'$ if the test statistic $T_W(S)\geq y_\alpha,$ where critical value $y_\alpha$ is the minimum positive number such that 
\begin{align}
 \nonumber   &\sum_{\theta\in \Theta_0: |I_\theta^c|\leq \frac{\sigma^2 |W|}{y_\alpha M_0}}2e^{- \frac{y_\alpha^2}{4 m^3 \sigma^2 }\frac{|I^c_\theta||I_\theta|}{|W|}}\\
 \label{critalv}    &+\sum_{\theta\in \Theta_0: |I_\theta^c|> \frac{\sigma^2 |W|}{y M_0}}2e^{- \frac{y_\alpha}{2 M_0 m^3}|I_\theta|+\frac{\sigma^2 |W|}{4 M_0^2 m^3 |I_\theta^c|}|I_\theta|}\leq \alpha.
\end{align}
\begin{remark}
In the case of Gaussian random field $\{\xi_k,k\in \Z^D\},$ we can obtain the same statements of Corollaries \eqref{cor2} and \eqref{cor3} with $M_0=\sigma.$ Indeed, $\E |\xi_0|^p\leq \sigma^{p-2}\sigma^{2}\E|\zeta|^{p-2},$ $p\geq 2,$ where $\zeta\sim N(0,1).$ So we put $H=M_0$ in \eqref{thm:c}.
\end{remark}
Simplifying the result of Corollary \ref{cor3}, we get that  \eqref{cor3eq1} and \eqref{cor3eq2} is bounded by 
\begin{align*}
    &2\left|\left\{\theta\in \Theta_0: |I_\theta^c|\leq \frac{\sigma^2 |W|}{y M_0}\right\}\right|e^{- \frac{y^2}{4 m^3 \sigma^2 }|W|\gamma_0(1-\gamma_0)}\\
    &+2\left|\left\{\theta\in \Theta_0: |I_\theta^c|> \frac{\sigma^2 |W|}{y M_0}\right\}\right|e^{- \frac{y}{4 M_0 m^3}|W|\gamma_0}.
\end{align*}

Particularly, if $y<\frac{\sigma^2}{M_0(1-\gamma_0)}$ then
\begin{align}
\label{cor2:eq2}    \pr\left(\max_{\theta\in \Theta_0}|\eta(\theta)|\geq y\right)\leq 2\left|\Theta_0\right|e^{- \frac{y^2}{4 m^3 \sigma^2 }|W|\gamma_0(1-\gamma_0)},
\end{align}
and if $y>\frac{\sigma^2}{M_0(1-\gamma_1)}$ then
\begin{align}
\label{cor2:eq3}     \pr\left(\max_{\theta\in \Theta_0}|\eta(\theta)|\geq y\right)\leq 2\left|\Theta_0\right|e^{- \frac{y}{4 M_0 m^3}\gamma_0|W|}.
\end{align}

\subsection{Change-point detection in simulated random fields}
In this section, we study the behaviour of the test statistics $T_W(S)$ given in \eqref{teststat} and probabilities of 1st-type error $\pr_{H_0'}(T_W(S)\geq y_\alpha)$ with respect to different values of $\sigma^2$ and $m.$

The form of $T_W$ allows to test the existence of the anomaly regions of arbitrary form and arbitrary number of connected components. On the other hand, we need to decrease the value $|\Theta|$ up to a feasible quantity for computational reasons (see bounds \eqref{cor2:eq2} and \eqref{cor2:eq3}).
Let $W=[1,M_1]\times[1,M_2]\times [1,M_3]\cap \N^3.$ We fix $\gamma_0=0.05$ and $\gamma_1=0.5,$ as the anomaly should not cover the majority of the window. In this paper we restrict $I_\theta$ to be a single rectangular parallelepiped of the form $[1+\Delta_0 i_1,1+\Delta_0 i_1+\Delta_1 l_1] \times [1+\Delta_0 i_2,1+\Delta_0 i_2+\Delta_1 l_2] \times [1+\Delta_0 i_3,1+\Delta_0 i_3+ \Delta_1 l_3]$. Then the parametric set of significant anomaly regions is given by 
\begin{align}
\nonumber    \Theta_0=&\left\{ \left(1+\Delta_0 i_1,1+\Delta_0 i_2,1+\Delta_0 i_3,\Delta_1 l_1,\Delta_1 l_2,\Delta_1 l_3\right),\right.\\
\nonumber &\text{ where } (i_1,i_2,i_3,l_1,l_2,l_3)\in \N_0^6,\\
\label{parset}    &
    1+\Delta_0i_j+\Delta_1 l_j\leq M_j, l_j\geq L_M, 1\leq j\leq 3,\\
\nonumber    &\left. \gamma_0\leq\frac{|I_{\theta}\cap J|}{M_1 M_2 M_3}\leq\gamma_1\right\}.
\end{align}
The offset parameters $\Delta_0$ and $\Delta_1$ as well as the parameter $L_M$ controlling the minimal edge length of the cuboids have to be chosen by the user.

Assuming the $m-$dependence for our observations, we do not know the exact value of $m.$ Hence, we need to impose some restrictions on the field $\xi$. First, if we know a-priori the maximum length of a typical fiber we can immediately obtain the bound for $m.$ Second, we can estimate the covariance function of the random field $\{s_k,k\in W\}$ and assess $m$ as the range when this empirical covariance is sufficiently close to zero. 
From relation \eqref{cor2:eq3} with  $\frac{\sigma^2}{M_0(1-\gamma_1)}<y<M$ we obtain the following approximate bound for an admissible $m:$
\begin{equation}
\label{mbound}
\begin{gathered}2\left|\Theta_0\right|\exp\left(- \frac{y}{4 M m^3}\gamma_0|W|\right) \leq \alpha \\
\Rightarrow m^3\leq \frac{\gamma_0}{4\log (2|\Theta_0|/\alpha)}|W|.
\end{gathered}
\end{equation}
For example, for $|\Theta_0|=10^4,$ $\alpha=0.05,$ $\gamma_0=0.05,$ one gets 
\begin{equation}
    \label{msize} m\leq\frac{1}{10}|W|^{1/3} \text{ or } |W|\geq 10^3 m^3.
\end{equation}

Let us now compare the empirical probability of the error of the 1-st type  with the bounds \eqref{critalv} for $\pr_{H_0'}(T_W(\cdot)\geq y_\alpha).$ 
We generate 300 realizations of a Gaussian centered $m$-dependent ($m=10$) random field $\{Y_k,k\in W\}$ with {$W=[1,80]^3\cap \N^3$} (which is matched to the considered data sets) and $Y_k\sim N(0,1).$ The dependence is modelled as follows: random variables  
$Y_{1+m k}, k\geq 0$ are independent, and $Y_{1+m k}=Y_{l+m k},k\in \N^3$ for all $l\in\{1,\ldots,m\}^3.$ We take  $\Delta_0=\Delta_1=8,$ $\gamma_0=0.05,$ and $\gamma_1=0.5.$ In this case, $|\Theta_0|=11954.$ Based on the simulated sample of values of the test statistics $T_W(Y)$ we compute the empirical critical value $\hat{y}_\alpha=0.6396$ for $\alpha=0.05.$ From comparison of $\hat{y}_\alpha$ with critical values $y_\alpha$ based on inequality \eqref{critalv} with $M_0=\sigma$ (presented in Table \ref{tb1}), we see that even under the exact value of $\sigma^2=1,$ critical values $y_\alpha$ are quite conservative. For example,   $y_\alpha=0.7198$ for  $m=7$ is still greater than $\hat{y}_\alpha=0.6396$ generated for $m=10.$

Therefore, we can use critical values from inequality \eqref{critalv} with $m$ smaller than its real value.
\begin{table}[h]
    \caption{Critical values $y_{0.05}$ based on inequality \eqref{critalv} for different values of $m$ and $\sigma$.}
    \label{tb1}
    \begin{tabular}{crrrrr}
    \hline\noalign{\smallskip}
        $\sigma^2$&$m=10$&$m=9$&$m=8$&$m=7$&$m=6$\\
\noalign{\smallskip}\hline\noalign{\smallskip}
        1&1.0757&1.0492&0.8793& 0.7198& 0.5711\\
        4&2.1513&2.0985&1.7587 & 1.4394& 1.1423\\
        8&3.0424&2.9677&2.4871 & 2.0357 & 1.6154\\
        \noalign{\smallskip}\hline
    \end{tabular}
    \begin{tabular}{rrrr}
    \hline\noalign{\smallskip}
        $m=5$&$m=4$&$m=3$&$m=2$ \\
\noalign{\smallskip}\hline\noalign{\smallskip}
       0.4345& 0.3109& 0.2019& 0.1099  \\
        0.8690& 0.6218& 0.4039& 0.2198 \\
       1.2289& 0.8793& 0.5711& 0.3109\\
        \noalign{\smallskip}\hline
    \end{tabular}
\end{table}

\section{Cluster based anomaly detection}
For processes $\Xi$ of thick fibres introduced in Section 2, the evidence of an anomaly is tested by applying the test of Section 4 to the random field $\{s_k,k\in W\}$ of estimated local mean or entropy of the chosen fibre characteristic $w.$ In this paper, $w(x)$ is the average direction vector of the fibres of $\Xi$ at $x\in W$ or one of its coordinates  (introduced in Section 3 as $\hat{w}_i$).

Assume that the anomaly test presented in Section 4 rejected the hypothesis $H_0'$ (and hence $H_0$), i.e., we have an evidence of an anomalous fibre behaviour in the rectangular subregion $I_{\theta_0}$ of our image data. Now we are interested in a more accurate estimate of the geometry of this anomal.  The search for an anomaly region in a 3D image can be interpreted as a problem of splitting the volume of the image into two disjoint clusters: {\it homogeneous material} and {\it anomaly}. 

In our problem setting, the volume under investigation, $\bigcup_{\vec{l}\in J_W}W_{\vec{l}}$, is a union of scanning windows $W_{\vec{l}}$ with meaningful local direction information. Each of them yields the clustering attributes mean of local directions (MLD) and entropy. We need to classify all the windows $W_{\vec{l}}$ as either belonging to the homogeneous material or the anomaly. For this purpose, a spatial version of the Stochastic Expectation Maximization algorithm is used.

\subsection{Spatial modification of a Stochastic Approximation  Expectation Maximization (SAEM) algorithm}
\label{SSEM}
 We assume that under the alternative $H_1$ (see page \ref{hypothesis}), fibres in the material may have two different distributions $f_0$ and $f_1$ of local directions. Therefore, the distribution of clustering attributes is a mixture of the distributions $f_0$ and $f_1$.

The Expectation-Maximization algorithm (EM) is commonly used to separate modes in a finite mixture of distributions, cf. \cite{EM} for a review. It is an iterative procedure consisting of two steps: Expectation (Estimation) and Maximization. In general, one assumes that the probability law under study is a mixture of $k$ distributions from the same parametric family. In the first step, the hidden parameters of the sample distribution, i.e., the weights of the mixture components,  are estimated, while in the second step the resulting parameters are updated by maximizing the likelihood function.

Since the EM algorithm belongs to the so-called ``greedy'' algorithms, that is, it converges to the first local optimum that has been found, a modification that compensates this deficiency should be used. One way out is a random ``shaking'' of observations in each iteration. This method is the basis of {\it the Stochastic EM (SEM) algorithm} (cf. \cite{SEM,EM}). 

The SEM algorithm works relatively fast in comparison with other methods, and its results are non-sensitive to an initial approximation. Random perturbations on the parameter space in the S-step guarantee the convergence to the global maximum of the likelihood function 
and help to avoid unstable local maxima.
On the other hand, the outputs of the SEM algorithm form a Markov chain and the final solution is its stationary distribution. To avoid this additional problem we use a modification called 
SAEM (Stochastic Approximation of EM) algorithm which brings together advantages of both EM and SEM approaches, e.g. \cite{Celeux92}.

Assume that the observable distribution has a density of the form 
$$\varphi_{\vec{\delta}}(x)=\beta \varphi(x,\delta_1)+(1-\beta)\varphi(x,\delta_2),~x\in \R^d,$$
where $\varphi(x,\delta_i)$ is a multivariate Gaussian density with unknown parameter $\delta_i=(\mu_i, \Sigma_i),$ $i=1,2$ and $\beta\in [0,1].$ Here $\mu_i$ is the mean and $\Sigma_i$ is the covariance matrix of Gaussian component $i=1,2.$ The combined unknown parameter is $\vec{\delta}=(\beta,\delta_1,\delta_2).$ We call $\varphi(\cdot,\delta_1)$ and $\varphi(\cdot,\delta_2)$ the {\it first} and {\it the second component of the mixture}, respectively. 
For each observation $x_\mathbf{l},\mathbf{l}\in J_W,$ we define a new variable $y_\mathbf{l}=\mathbb{I}\{x_\mathbf{l}$ belongs to the first component$\}$. Therefore, we have two samples: observable $\mathbf{x}=\{x_\mathbf{l},\mathbf{l}\in J_W\}$ and unobservable $\mathbf{y}=\{y_\mathbf{l},\mathbf{l}\in J_W\}.$  Then the log-likelihood function equals 
\begin{align*}
    &\ln L(\vec{\delta},\mathbf{x},\mathbf{y})\\
    &=\sum_{\mathbf{l}\in J_W} \left[y_\mathbf{l}\ln(\beta \varphi(x_\mathbf{l},\delta_1))+(1-y_\mathbf{l})\ln((1-\beta) \varphi(x_\mathbf{l},\delta_2))\right]\\
    &=\nu_1 \ln \beta +\nu_2 \ln (1-\beta)+ \sum_{\mathbf{l}\in J_W:y_\mathbf{l}=1} \ln \varphi(x_j,\delta_1)\\
    &+\sum_{\mathbf{l}\in J_W:y_\mathbf{l}=0} \ln \varphi(x_j,\delta_2),
\end{align*}
where $\nu_1=\sum_{\mathbf{l}\in J_W}y_\mathbf{l}$ denotes the number of observations belonging to the first mixture component and $\nu_2=m_1 m_2 m_3-\nu_1$ observations belong to the second one.

 Assume that we know the a posteriori probability $q_{\mathbf{l}}^{(k-1)}, \mathbf{l}\in J_W,$  that $x_\mathbf{l}$ belongs to the first component $\varphi(\cdot,\delta_1),$ where $k-1$ is the iteration number. Let us describe the EM part.
During the M-Step we obtain new estimates of the parameters $\hat{\delta}^{(k)}_1=(\mu^{(k)}_1,\Sigma^{(k)}_1),$ $\hat{\delta}^{(k)}_2=(\mu^{(k)}_2,$ $\Sigma^{(k)}_2),\hat{\beta}^{(k)}$
by \begin{equation}
\label{EMMstep1}
\begin{gathered}
    \mu^{(k)}_1=\frac{\sum_{\mathbf{l}\in J_W} q^{(k-1)}_\mathbf{l}x_\mathbf{l}}{\sum_{\mathbf{l}\in J_W} q^{(k-1)}_\mathbf{l}},\\
    \Sigma^{(k)}_1=\frac{\sum_{\mathbf{l}\in J_W}q^{(k-1)}_\mathbf{l}(x_\mathbf{l}-\mu^{(k)}_1)(x_\mathbf{l}-\mu^{(k)}_1)^T }{\sum_{\mathbf{l}\in J_W} q^{(k-1)}_\mathbf{l}},
    \end{gathered}
    \end{equation}
\begin{equation}
    \label{EMMstep2}
    \mu^{(k)}_2=\frac{\sum_{\mathbf{l}\in J_W} (1-q^{(k-1)}_\mathbf{l}) x_\mathbf{l}}{\sum_{\mathbf{l}\in J_W}(1- q^{(k-1)}_\mathbf{l})},
\end{equation}
\begin{equation}
    \label{EMMstep3}
    \Sigma^{(k)}_2=\frac{\sum_{\mathbf{l}\in J_W}(1-q^{(k-1)}_\mathbf{l})(x_\mathbf{l}-\mu^{(k)}_2)(x_\mathbf{l}-\mu^{(k)}_2)^T }{\sum_{\mathbf{l}\in J_W} (1-q^{(k-1)}_\mathbf{l})},
\end{equation}
\begin{equation}
\label{EM est:alpha}
\hat{\beta}^{(k)}=\frac{1}{m_1 m_2 m_3}\sum_{\mathbf{l}\in J_W}q^{(k-1)}_\mathbf{l}.
\end{equation}

In the E-step we compute the new probabilities based on \eqref{EMMstep1}-\eqref{EM est:alpha} as
\begin{equation}
\label{qEM}
q_{\mathbf{l}}^{k,EM}=\frac{\hat{\beta}^{(k)} \varphi(x_\mathbf{l},\hat{\delta}^{(k)}_1)}{\hat{\beta}^{(k)} \varphi(x_\mathbf{l},\hat{\delta}^{(k)}_1)+(1-\hat{\beta}^{(k)}) \varphi(x_\mathbf{l},\hat{\delta}^{(k)}_2)},~ \mathbf{l}\in J_W.
\end{equation}

In the SEM-part we act in a different way. In the S-step we generate independent Bernoulli-distributed random variables $y_{\mathbf{l}}^{(k)}\in\{0,1\}$ with probabilities\linebreak $\pr(y_\mathbf{l}^{k}=1)=q_{\mathbf{l}}^{(k-1)},$ $\mathbf{l}\in J_W$. 

During the M-Step we get $\nu_1^{(k)}=\sum_{\mathbf{l}\in J_W} y_{\mathbf{l}}^{(k)}$ and the estimates $\hat{\delta}^{(k)}_1=(\mu^{(k)}_1,\Sigma^{(k)}_1),$  $\hat{\delta}^{(k)}_2=(\mu^{(k)}_2,\Sigma^{(k)}_2),$  $\hat{\beta}^{(k)}$
by \begin{equation}
    \label{SEMMstep1}
    \begin{gathered}
    \mu^{(k)}_1=\frac{\sum_{\mathbf{l}\in J_W: y_\mathbf{l}^{(k)}=1} x_\mathbf{l}}{\nu^{(k)}_1},\\
    \Sigma^{(k)}_1=\frac{\sum_{\mathbf{l}\in J_W: y_\mathbf{l}^{(k)}=1}(x_\mathbf{l}-\mu^{(k)}_1)(x_\mathbf{l}-\mu^{(k)}_1)^T }{\nu^{(k)}_1 },
    \end{gathered}
    \end{equation}
\begin{equation}
    \label{SEMMstep2}
    \begin{gathered}
   \mu^{(k)}_2=\frac{\sum_{\mathbf{l}\in J_W: y_\mathbf{l}^{(k)}=0} x_\mathbf{l}}{\nu^{(k)}_2},\\
    \Sigma^{(k)}_2=\frac{\sum_{\mathbf{l}\in J_W: y_\mathbf{l}^{(k)}=0}(x_\mathbf{l}-\mu^{(k)}_2)(x_\mathbf{l}-\mu^{(k)}_2)^T }{\nu^{(k)}_2 },
    \end{gathered}
\end{equation}
\begin{equation}
\label{SEM est:alpha}
\hat{\beta}^{(k)}=\frac{\nu^{(k)}_1}{m_1 m_2 m_3}=\frac{1}{m_1 m_2 m_3}\sum_{\mathbf{l}\in J_W}{y^{(k)}_\mathbf{l}}.
\end{equation}
In the E-Step, we compute the updated probabilities $q_{\mathbf{l}}^{(k,SEM)}$  based on \eqref{SEMMstep1}-\eqref{SEM est:alpha} by relation \eqref{qEM}.

The essential idea of the SAEM algorithm is to mix $q_{\mathbf{l}}^{(k,EM)}$ and $q_{\mathbf{l}}^{(k,SEM)}$ in iteration step $k$ as 
\begin{equation}
\label{q}
q_{\mathbf{l}}^{(k)}=\lambda_k q_{\mathbf{l}}^{(k,SEM)}+(1-\lambda_k) q_{\mathbf{l}}^{(k,EM)}, \mathbf{l}\in J_W,
\end{equation}
which gives the a posteriori probabilities for the next $(k+1)$th iteration. Here in \eqref{q}, $\{\lambda_k,k\ge 1\}$ is a sequence of positive real numbers $\lambda_k\in (0,1)$ decreasing to zero. We stop the SAEM algorithm after the $k-$th iteration if $\sum_{\mathbf{l}\in J_W}|q_{\mathbf{l}}^{(k-1)}-q_{\mathbf{l}}^{(k)}|\leq \varepsilon,$ where $\varepsilon$ is some threshold.
In the following, the choice of parameters is a result of experimental tuning to our image data yielding good practical results. Particularly, we use $\lambda_k=\frac{50}{50+k^2},k\geq 1$ and $\varepsilon=0.0001$ in our computations.

When the SAEM algorithm stops in the $k_0$-th iteration, we obtain the values 
$\{q^{(k_0)}_\mathbf{l},\mathbf{l}\in J_W\}$ which indicate that $x_\mathbf{l}$ belongs to the first component if $q^{(k_0)}_\mathbf{l}>1/2$
and to the second one, otherwise.

Applying the above SAEM algorithm to our image data yields diffuseness in the resulting clusters (see Figure \ref{fig:Picture13l}).
To avoid this, we propose a smoothing modification ({\it Spatial SAEM}), which takes the spatial location of the sample data into account. Let us describe the new {\it Spatial step}.

{
Let SAEM stop after $k_0$ iterations. For each sample entry $x_\mathbf{l},\mathbf{l}=(l_1,l_2,l_3)\in J_W$
we define the coordinate  $v_\mathbf{l}=\left((l_1-1)M\Delta,(l_2-1)M\Delta,(l_3-1)M\Delta\right),$ that is a vertex of the cube $W_\mathbf{l}.$  
In each further iteration, i.e. for $k > k_0$, Bernoulli random variables $y^{(k)}_\mathbf{l}\in \{0,1\},\mathbf{l}\in J_W$ with success probability $q^{(k_0)}_\mathbf{l},\mathbf{l}\in J_W$ are simulated. Now  $y^{(k)}_\mathbf{l}$ classifies $x_\mathbf{l}$ in such a way that  $y^{(k)}_\mathbf{l}=0$ indicates that $x_\mathbf{l}$ belongs to the first component and the second one, otherwise. 
Then, we compute the number $a^{(k)}_\mathbf{l}$ of neighbors of $x_\mathbf{l}$ belonging to the same cluster as $x_\mathbf{l}$. Neighborhood is defined in terms of the $r-$neighborhood of $v_\mathbf{l}$ such that
\begin{equation}
\label{defa}a^{(k)}_\mathbf{l}=\sum\limits_{\mathbf{i}\in J_W, \mathbf{i}\neq \mathbf{l}}\mathbb{I}(y^{(k)}_\mathbf{i}=y^{(k)}_\mathbf{l},\|v_\mathbf{l}-v_\mathbf{i}\|_{\infty}\leq r),
\end{equation}
$r> 0.$
}

If all $a^{(k)}_\mathbf{l}$ are greater than or equal to a certain threshold $a$ (for our image data, $a=3$ is used) we call the classification $y^{(k)}_\mathbf{l} ,\mathbf{l}\in J_W$,{ \it  admissible} and move to the next iteration. Otherwise, for sample entries $x_\mathbf{l}$ with $a^{(k)}_\mathbf{l}$ being less than $a,$ we change 
$y^{(k)}_\mathbf{l}$, hence, the class of $x_\mathbf{l}$. If the new set $y^{(k)}_\mathbf{l} ,\mathbf{l}\in J_W$, is admissible, then we pass to iteration $k+1.$ If no, we resimulate $y^{(k)}_\mathbf{l} ,\mathbf{l}\in J_W$, until $a^{(k)}_\mathbf{l}\geq a$ for all $\mathbf{l}\in J_W.$

The smoothing procedure stops when $K$ admissible classifications have been generated ($K=1000$ is used). The final a posteriori probabilities 
$q_\mathbf{l}$ in the space of all admissible classifications are computed over the sample of 
$\{y^{(k)}_\mathbf{l} ,\mathbf{l}\in J_W\},k_0\leq k \leq k_0+K$ by 
\begin{equation}
    \label{qf}
    q_\mathbf{l} = \frac{1}{K}\sum_{k=k_0}^{k_0+K} y^{(k)}_\mathbf{l},~ \mathbf{l}\in J_W.
\end{equation}

We also get estimates of components' weights $(\hat{\beta},1-\hat{\beta})$ by \eqref{EM est:alpha}. If $\hat{\beta}\geq  0.5$ we say that the second component corresponds to the ``anomaly''. The observation window $W_\mathbf{l}$ thus belongs to the zone of homogeneous material if $q_\mathbf{l}\geq 0.5$ and to the anomaly zone if $q_\mathbf{l}<0.5.$ If $\hat{\beta}< 0.5$ the roles of the components are swapped.

\section{Application to 3D image data of fibre materials}\label{sec:AnomalyDetection3DImages}

\begin{figure*}[h]
\center{\includegraphics[height=0.45\textheight ]{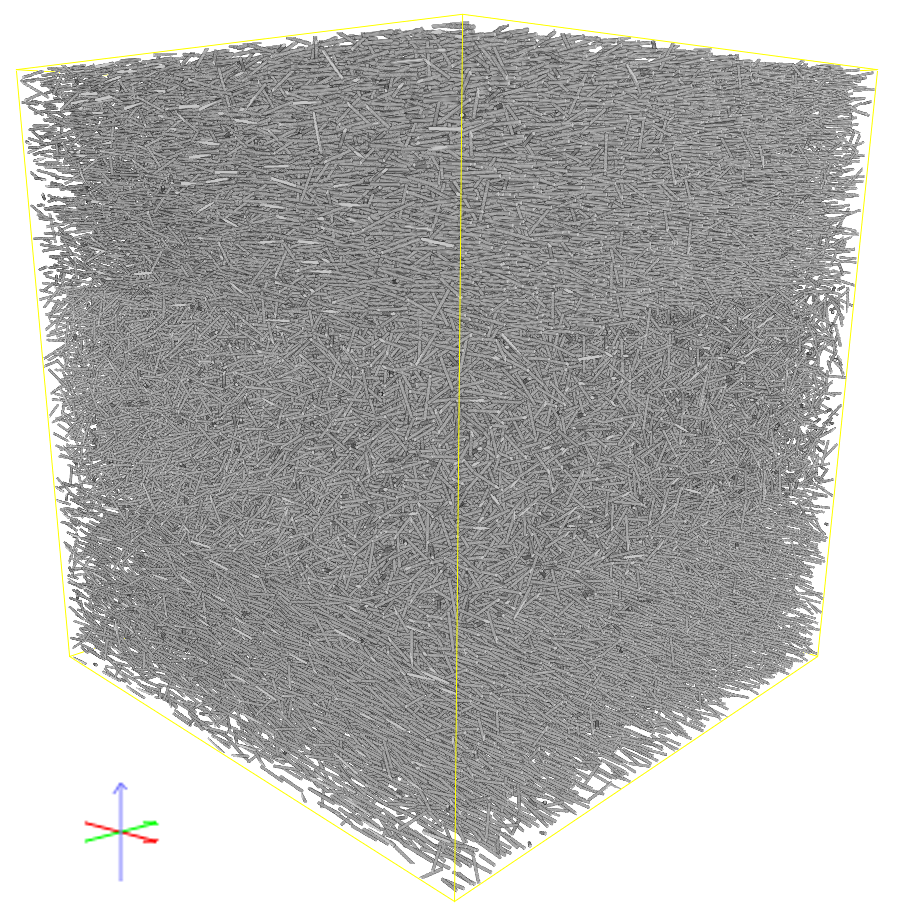}}
		\caption{Visualization of simulated layered RSA fibre data, $2000\times2000\times2100$ voxels}
		\label{fig:Picture5l}
\end{figure*}
\begin{figure*}[h]
\center{\includegraphics[height=0.45\textheight ]{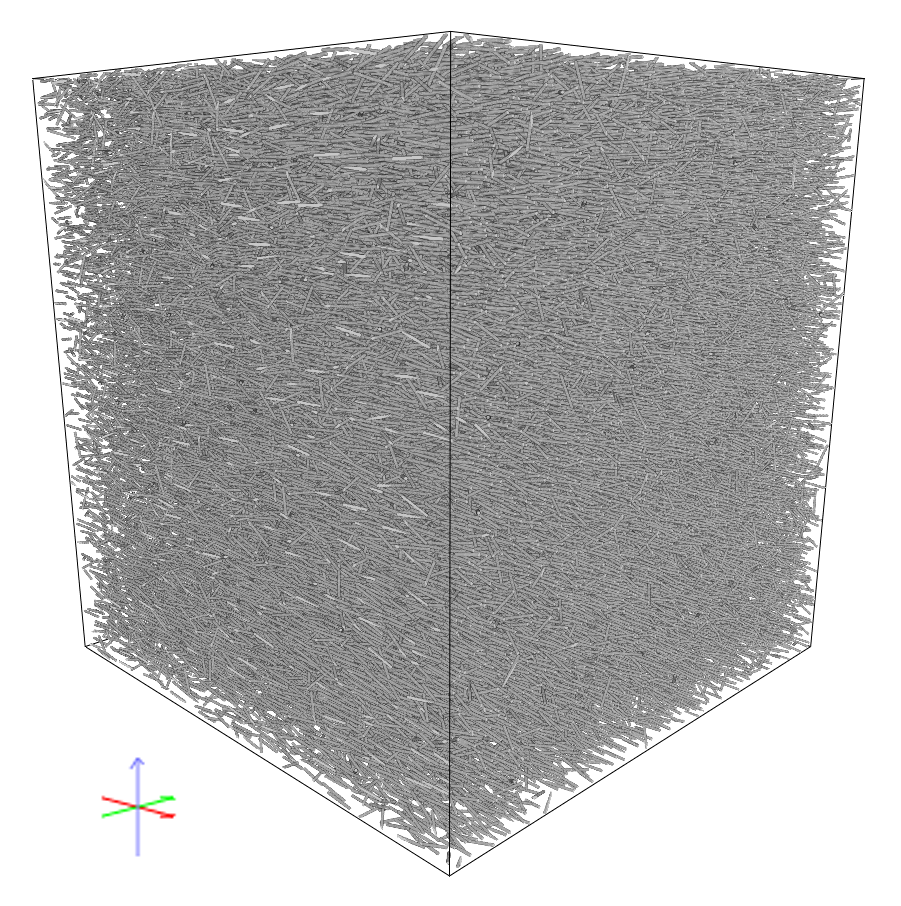}}
		\caption{Visualization of simulated homogeneous RSA fibre data, $2000\times2000\times2100$ voxels}
		\label{fig:Picture5r}
\end{figure*}

\subsection{Simulated data}
First, we illustrate the use of the methods from Sections 4 and 5 on simulated 3D fibre images. We choose a {\it random sequental absorbtion} (RSA) model that randomly adds fibres to the existing material, such that they do not intersect each other, cf. \cite{andra2014geometric,Red11}.  Figure~\ref{fig:Picture5l} shows simulated RSA fibre data in an image of $2000\times2000\times2100$ voxels. The sample exhibits three layers, where fibres differ in their local directional distribution. Each layer has a thickness of 700 voxels and contains 82474 fibres with a constant radius of 4 voxels and length of 100 voxels. Fibre directions are distributed according to a special case of the Angular Central Gaussian distribution described in \cite{Franke}. In the two outer layers, the preferred direction is the $x$-direction and the concentration parameter is $\beta=0.1$ resulting in a high concentration of the fibres along the main direction. In the middle layer, which is considered the anomaly region, the preferred direction is the $y$-direction and the fibres are less concentrated ($\beta=0.5$).

Additionally, we investigate a homogeneous RSA data set where no anomalies should be detected. The data set consists of an image of $2000\times2000\times 2100$ voxels. Here, the concentration parameter of the fibre direction distribution is $\beta=0.1$ in the whole sample. The preferred direction is the $x$-direction.  The fibre radius is 4 voxels, the fibre length is 100 voxels. A visualisation of a realisation of this model is shown in Figure~\ref{fig:Picture5r}.

Now we apply the change-point analysis of Section 4 to random fields of mean local directions and entropy estimates for the homogeneous and layered RSA data.

{To do so,  we transform the data of average local directions $X_\mathbf{k},\mathbf{k}\in J$. In order to avoid cancelling effect of averaging,  we build for each coordinate $x,y,z,$ the samples $\tilde{x}_\mathbf{k},\tilde{y}_\mathbf{k},\tilde{z}_\mathbf{k},$ such that their entries lie in the hemispheres $\mathbb{S}_x^2=\{(x,y,z)\in \mathbb{S},x\geq 0\},\mathbb{S}_y^2=\{(x,y,z)\in \mathbb{S},y\geq 0\},\mathbb{S}_z^2=\{(x,y,z)\in \mathbb{S},z\geq 0\},$ respectively, i.e., $\tilde{x}_\mathbf{k}=|x|_\mathbf{k},$ $\tilde{y}_\mathbf{k}=|y|_\mathbf{k},$ $\tilde{z}_\mathbf{k}=|z|_\mathbf{k},$}

{The sample of the estimated entropy values $\hat{E}_\mathbf{k}$ of directional distribution of fibres $X_\mathbf{i},\mathbf{i}\in J$ in the windows $W_\mathbf{k},\mathbf{k}\in J_W$ is build by estimator \eqref{dobr_est_v} over the transformed directions $\hat{X}_\mathbf{i}\in \mathbb{S}_+^2,\mathbf{i}\in J.$}

We apply the results of Section 4 consequently to the random fields $s_{\mathbf{k}}=$ $\tilde{x}_\mathbf{k},\tilde{y}_\mathbf{k},$ or $\tilde{z}_\mathbf{k},$ and  $\hat{E}_\mathbf{k}.$
Therefore, we have the following 4 pairs of hypotheses of $(H_0',H_1')$-type.
\begin{itemize}
    \item $H_0^x:$  $\E \tilde{x}_\mathbf{k}=\mu_x$ for every $\mathbf{k}\in W$ vs.
    \item $H_1^x:$  $\exists \theta_0\in \Theta_0$ such that $\E \tilde{x}_\mathbf{k}=\mu_x+h_x, \mathbf{k}\in I_{\theta_0}$ and $\E \tilde{x}_\mathbf{k}=\mu_x,\mathbf{k}\in I_{\theta^c_0},$ $h_x\neq0$;
    \item $H_0^y:$  $\E \tilde{y}_\mathbf{k}=\mu_y$ for every $\mathbf{k}\in W$ vs.
    \item $H_1^y:$  $\exists \theta_0\in \Theta_0$ such that $\E \tilde{y}_\mathbf{k}=\mu_y+h_y, \mathbf{k}\in I_{\theta_0}$ and $\E \tilde{y}_\mathbf{k}=\mu_y,\mathbf{k}\in I_{\theta^c_0},$ $h_y\neq0$;
    \item $H_0^z:$  $\E \tilde{z}_\mathbf{k}=\mu_z$ for every $\mathbf{k}\in W$ vs.
    \item $H_1^z:$  $\exists \theta_0\in \Theta_0$ such that $\E \tilde{z}_\mathbf{k}=\mu_z+h_z, \mathbf{k}\in I_{\theta_0}$ and $\E \tilde{x}_\mathbf{k}=\mu_z,\mathbf{k}\in I_{\theta^c_0},$ $h_z\neq0$;
     \item $H_0^E:$  $\E \hat{E}_\mathbf{k}=\mu_E$ for every $\mathbf{k}\in J_W$  vs.
    \item $H_1^E:$  $\exists \theta_0\in \Theta_0$ such that $\E \hat{E}_\mathbf{k}=\mu_E+h_E, \mathbf{k}\in I_{\theta_0}$ and $\E \hat{E}_\mathbf{k}=\mu_E,\mathbf{k}\in I_{\theta^c_0},$ $h_E\neq0.$
\end{itemize}
Since we test only 4 hypotheses simultaneously, we stick to the classical Bonferroni method, e.g. we test each direction and entropy separately with significance level $\frac{1}{4}\alpha.$ 

Before running the algorithms we need to choose the right size of scanning windows. From the initial layered  and homogeneous RSA images with $2000\times2000\times2100$ voxels we obtain $83\times83\times 87$ small windows $\widetilde{W}_{\vec{i}}$ with $24\times24\times24$ voxels each, and 463537 and 460559 nonempty entries, respectively.

Due to the model parameters (fibre length of 100 voxels corresponds to 5 points in $W$), we can assume the random field $\tilde{X}_\mathbf{k}$ to be $m-$dependent { with $m=5$ and $\sigma^2=0.2, M_0=\frac{1}{2}.$ For mean local directions $\tilde{x},\tilde{y},$ and $\tilde{z}$, the parametric set $\Theta_0$ is constructed with $\Delta_0=\Delta_1=8,$ $\gamma_0=0.05,\gamma_1=0.5,$ and $L_M=22$ in \eqref{parset}, $|\Theta_0|=39395.$}

We point out that the samples $(\tilde{x},\tilde{y},\tilde{z})$ and $\hat{E}$ have different sizes due to the construction described in Section 3. Therefore, the parameters $m,\sigma^2$ and the parameter set $\Theta_0$ in \eqref{parset} for $\hat{E}$ differ from the ones for $\tilde{x},\tilde{y},\tilde{z}.$

{For the sample of estimated entropy $\hat{E}_{\mathbf{k}},\mathbf{k}\in J_W,$ we have $m=1, \sigma^2=0.5$ and the parametric set $\Theta_0$ is constructed with $\Delta_0=\Delta_1=2,$ $\gamma_0=0.05,\gamma_1=0.5,$ and $L_M= 4$ in \eqref{parset}, $|\Theta_0|=16536.$} Entropy values approximately have a normal distribution, so we put $M_0=\sigma$ in \eqref{critalv}.

The computed statistics (given in \eqref{teststat}) $T_W(\tilde{x}),T_W(\tilde{y}),$ $T_W(\tilde{z}),$ $T_W(\hat{E})$  and corresponding $p$-values from relation \eqref{critalv} are presented in Table \ref{tb2} for the homogeneous and in Table \ref{tb3} for the layered RSA data. Thus, there is no evidence to reject $H^x_0,$ $H^y_0,$ $H^z_0,$ $H^E_0$ in the homogeneous case. The described test allows to claim that there is an anomaly region in the layered RSA image data, {because we reject $H^x_0,$ $H^y_0,$ and $H^E_0,$ but have no evidence to reject $H^z_0.$}

\begin{table}[h]
    \centering
    \begin{tabular}{cr r r r }
    \hline\noalign{\smallskip}
        Attribute &  Sample var. & Test statistics & $p-$value \\
        \noalign{\smallskip}\hline\noalign{\smallskip}
        $\tilde{x}$& 0.04360 & 0.0344 &  1.00\\
        $\tilde{y}$& 0.03743 & 0.0130 &1.00\\
        $\tilde{z}$& 0.03749 & 0.0146&1.00\\
        $\widetilde{E}$& 0.08984 &0.0942 &1.00\\
        \noalign{\smallskip}\hline
    \end{tabular}
    \caption{Change-point test for mean local directions of homogeneous RSA data.}
    \label{tb2}
\end{table}

\begin{table}[h]
    \centering
    \begin{tabular}{cr r r }
\hline\noalign{\smallskip}
        Attribute &  Sample var. & Test statistics & $p-$value \\
        \noalign{\smallskip}\hline\noalign{\smallskip}
        $\tilde{x}$&0.10592 & 0.44036 & $4.6\times 10^{-30}$\\
        $\tilde{y}$&0.10948 &0.43163 &$2.8\times 10^{-23}$\\
        $\tilde{z}$& 0.06151 &0.18764 &0.301 \\
        $\widetilde{E}$& 0.3583 &1.07030 &0.00\\
        \noalign{\smallskip}\hline
    \end{tabular}
    \caption{Change-point test for mean local directions of layered RSA data.}
    \label{tb3}
\end{table}

Therefore, the choice of the mean of local directions attribute for the change-point analysis in our problem with layered fibre image data is reasonable. Moreover, depending on the data (e.g. containing whirlpools of fibres) it may be better to choose entropy or other attributes to test for other types of anomalies.  

It follows from \cite[Theorem 2.4.]{Penrose} that the Dobrushin estimator of the entropy of i.i.d. vectors on a $C^1-$smooth manifold is asymptotically Gaussian. 
Although the RSA fibre data do not satisfy the i.i.d. assumption of mutual independence of fibre locations and directions, the estimated local directional entropy $\widehat{E}$ for the homogeneous data seems to have a unimodal distribution, see Figure~\ref{fig:Picture11r}. 
Assuming that the Gaussian distribution provides a reasonable approximation also in this case, we apply the $3\sigma$-rule with $\sigma^2$ being the sample variance of $\widehat{E}$, compare \cite{ruiz2016entropy}, to find anomaly regions in both Figures~\ref{fig:Picture5l} and \ref{fig:Picture5r}. 

\begin{figure}[h]
	\center{\includegraphics[width=0.6\linewidth,keepaspectratio]{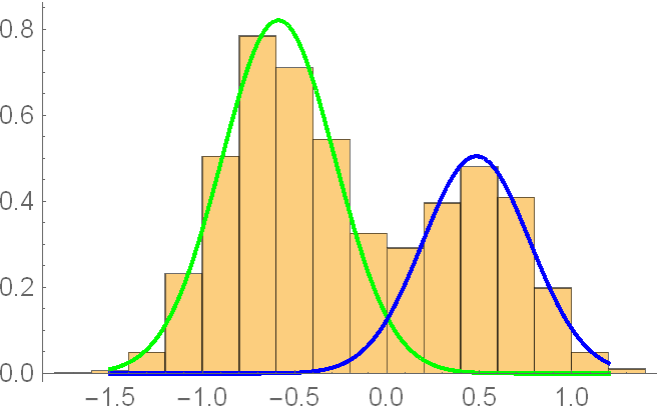}}
	\caption{Histogram of frequencies of the local entropy of fibre directions and two separated Gaussian probability density functions found by the spatial SAEM algorithm. Layered RSA fibre data, $\sigma=0.5986$.}
	\label{fig:Picture11l}
\end{figure}
\begin{figure}[h]	
\center{\includegraphics[width=0.6\linewidth,keepaspectratio]{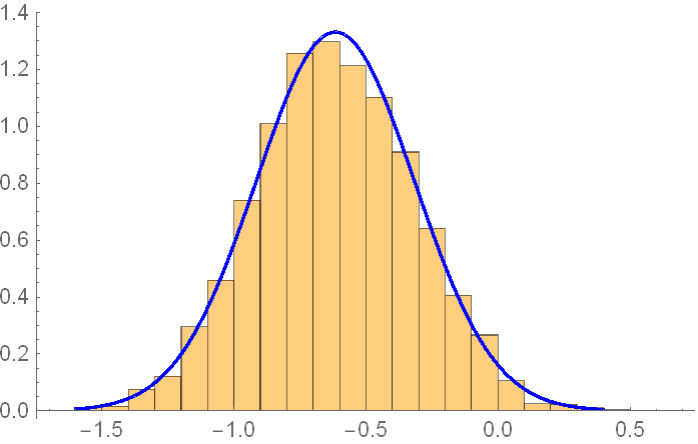}}
	\caption{Histogram of frequencies of the local entropy of fibre directions fitted Gaussian probability density. Homogeneous RSA fibre data, $\sigma=0.2997$}
	\label{fig:Picture11r}
\end{figure}

One can see that all centered entropy values lie in the interval $[-3\sigma; 3\sigma],$ so the $3\sigma-$method does not distinguish between homogeneous (Figure~\ref{fig:Picture5r}) and inhomogeneous (Figure~\ref{fig:Picture5l}) images.
We conclude that the $3\sigma$-rule for anomaly detection does not work well if the anomaly regions are large enough to produce histograms of the clustering attribute with many modes.

The fact that the distribution of $\widehat{E}$ seems to have two modes might indicate that it is a mixture of two Gaussian distributions. So we apply the Spatial SAEM algorithm from Section~\ref{SSEM} to separate these modes. 
By an empirical study, a scanning window $W$ consisting of $5\times5\times5$ small windows $\widetilde{W}_{\vec{k}}$ was selected, i.e. $M=5.$  
Additionally, we put $r=\Delta M$ in \eqref{defa} by default.

First, let us consider clustering based on the attribute entropy. In the layered data, the Spatial SAEM algorithm finds two clusters and determines the distributional parameters for them. The clustering results are presented in Figurer~\ref{fig:Picture13l} and~\ref{fig:Picture13r} . Green labels denote the centers of scanning windows corresponding to the homogeneous fibre material and blue labels mark objects belonging to the anomalous region. As expected, the Spatial SAEM algorithm found no anomaly for homogeneous RSA fibre data. 
\begin{figure}[h]
	\center{\includegraphics[width=0.7\linewidth,keepaspectratio]{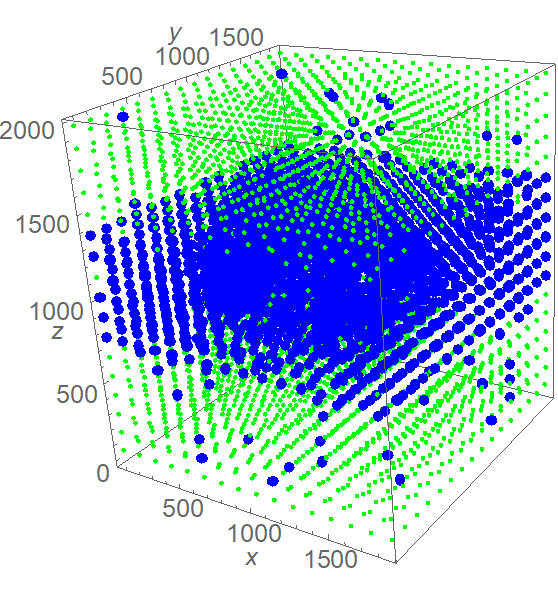} }
	\caption{Anomaly detection in layered RSA fibre data using the local entropy: SAEM algorithm.}
	\label{fig:Picture13l}
\end{figure}
\begin{figure}[h]
		\center{\includegraphics[width=0.7\linewidth,keepaspectratio]{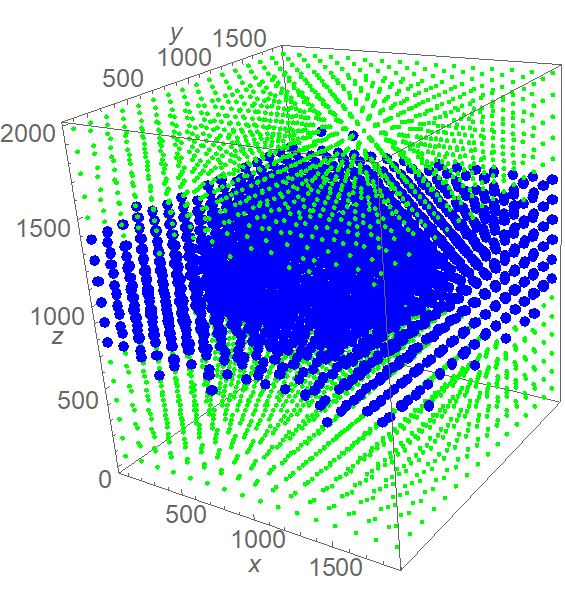}}
	\caption{Anomaly detection in layered RSA fibre data using the local entropy: Spatial SAEM algorithm.}
	\label{fig:Picture13r}
\end{figure}

We also ran the Spatial SAEM algorithm with the attributes mean of local direction (MLD) and a vector combining entropy and MLD. The results for the layered data are presented in Figures~\ref{fig:Picture14l} and~\ref{fig:Picture14r} . One can see that combination of both attributes gives a more reliable result.

\begin{figure}[h]
		\center{\includegraphics[width=0.7\linewidth,keepaspectratio]{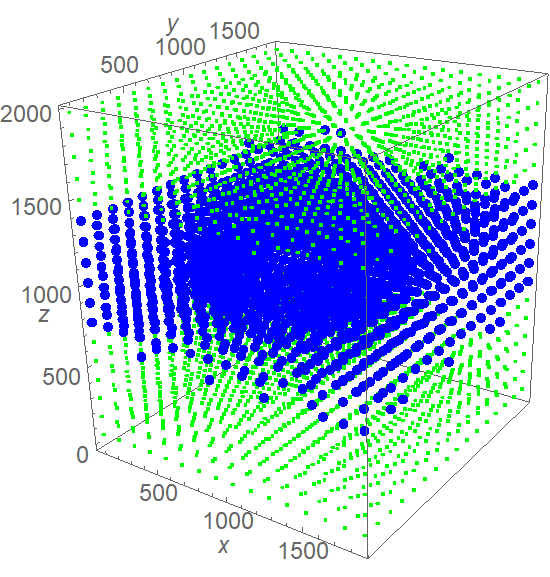}}
		\caption{Anomaly detection with Spatial SAEM algorithm in the layered RSA fibre data using mean of local directions.}
			\label{fig:Picture14l}
\end{figure}
\begin{figure}[h]
	\center{\includegraphics[width=0.7\linewidth,keepaspectratio]{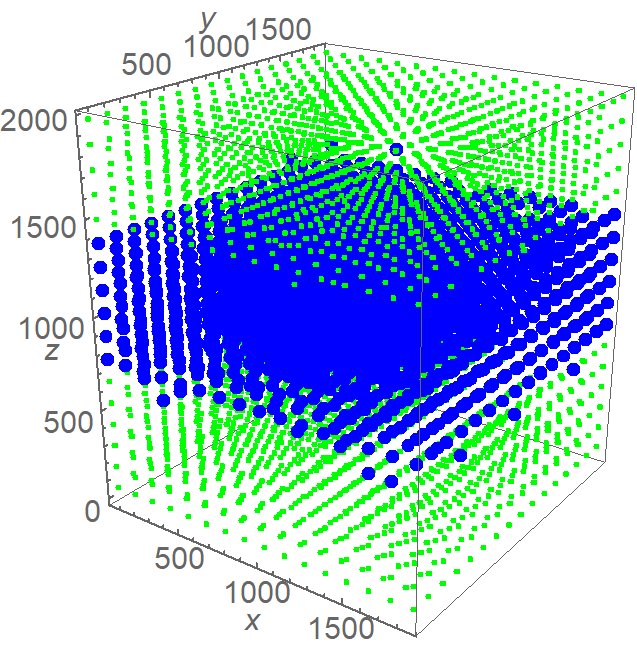}}
	\caption{Anomaly detection with Spatial SAEM algorithm in the layered RSA fibre data using local entropy and mean of local directions attributes.}
	\label{fig:Picture14r}
\end{figure}

\begin{remark}
The problem of clustering a fibre material into homogeneity and anomaly zones using vector-valued cluster attributes can be solved by a variety of other clustering methods, see the books \cite{everitetal11,ClustHandbook16,WKClus18} for an overview. In addition to the results reported here, we tried the recent AWC algorithm \cite{AWC}. However, the spatial SAEM approach yields better results, cf. \cite{ConfMat}. Moreover, it does not require a complex parameter tuning and operates fast.

We also tried to use a principal axis of fibre directions as a classification attribute in the described SAEM algorithm. But the results are worse than the ones for the MLD attribute. This effect can be explained by the fact that the distribution of principal axes has its support on a unit sphere and thus cannot be a mixture of Gaussian distributions in $\R^3.$ Therefore, the estimates in the M-step \eqref{EMMstep1}-\eqref{EMMstep3}, \eqref{SEMMstep1},\eqref{SEMMstep2} have to be modified, cf. \cite{Franke}. This task goes, however, beyond the scope of the present article.
\end{remark}

\subsection{Real glass fibre reinforced polymer}

Now we apply our anomaly detection approach to a 3D-image of a glass fibre reinforced polymer. The images are provided by the Institute for Composite Materials (IVW) in Kaiserslautern, see Figure~\ref{fig:Picture22l}. For a detailed description of the material we refer to \cite{IVWFibres}.

\begin{figure}[h]
		\center{\includegraphics[width=0.7\linewidth,keepaspectratio]{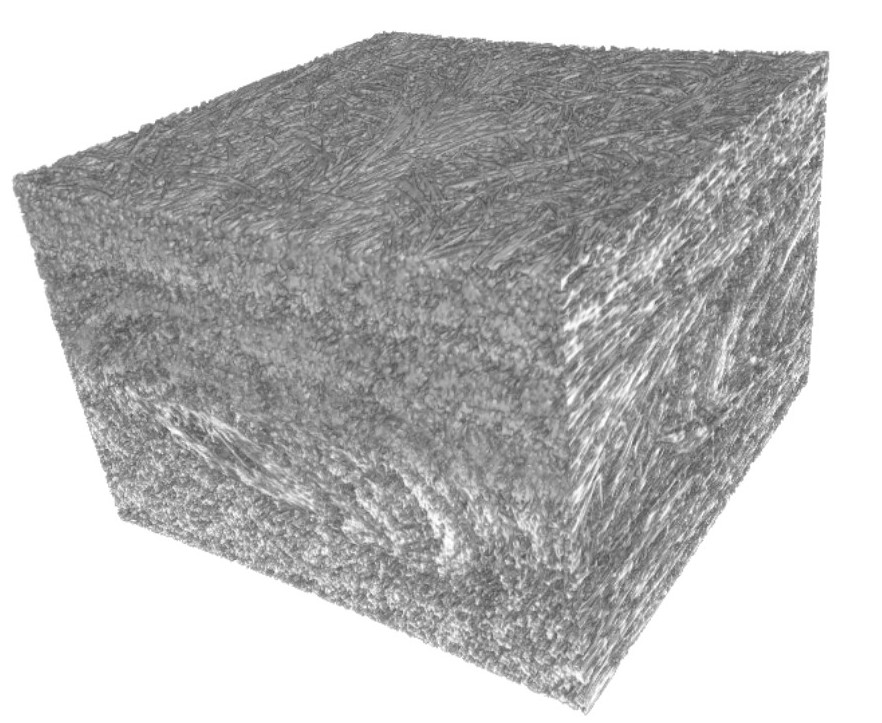}}
		\caption{3D image of a glass fibre reinforced composite material. 970$\times$1469$\times$1217 voxels, voxel spacing: 4 $\mu$m.}
		\label{fig:Picture22l}
\end{figure}
\begin{figure}[h]
	\center{\includegraphics[width=0.7\linewidth,keepaspectratio]{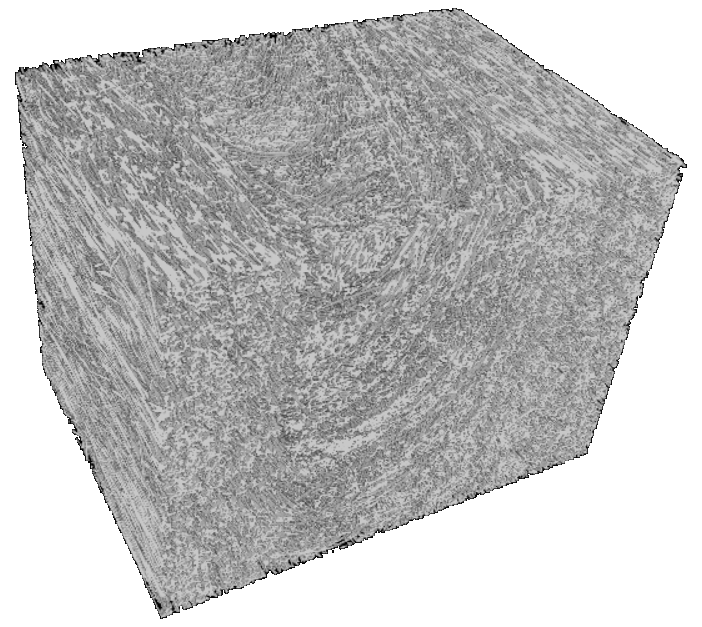}}
	\caption{3D image of a glass fibre reinforced composite material. The part of the data containing an anomaly (970 $\times$ 700 $\times$ 660 voxels) }
	\label{fig:Picture22r}
\end{figure}

We apply the change point analysis from Section 4 to  real data with $970\times1469\times1217$ voxels and the estimated radius of 3 voxels. We obtain $64\times97\times 80$ small windows $\widetilde{W}_{\vec{i}}$ with $15\times15\times15$ voxels.

{For mean local directions, the parametric set $\Theta_0$ is constructed with $\Delta_0=\Delta_1=8,$ $\gamma_0=0.05,\gamma_1=0.5$ and $L_M= 22$ in \eqref{parset}, which gives $|\Theta_0|=33004.$} To choose the suitable value of $m$ for $m-$dependence we need additional investigation. Under the hypotheses $H_0^x,H_0^y,H_0^z$ the random fields $\tilde{x}$,$\tilde{y}$, and $\tilde{z}$ are assumed to be stationary, so that we can estimate their covariance functions. We use the standard approach and estimate e.g. $\rho_x(\mathbf{h})=\mathbf{Cov}( \tilde{x}_{\mathbf{1}},\tilde{x}_{\mathbf{1}+\mathbf{h}})$ as
\begin{equation*}
\hat{\rho}(\mathbf{h})=\frac{1}{|K|-1}\sum_{\mathbf{k}\in K}  \left(\tilde{x}_{\mathbf{k}}-\bar{x}_0\right)\left(\tilde{x}_{\mathbf{k}+\mathbf{h}}-\bar{x}_\mathbf{h}\right),
\end{equation*}
where $K=\{ (k_1,k_2,k_3)\in \N^3, 1\leq k_1\leq M_1-h_1,1\leq k_2\leq M_2-h_2,1\leq k_3\leq M_3-h_3\},$ and $\bar{x}_0$ and $\bar{x}_\mathbf{h}$ are the sample means of $\tilde{x}$ over the index ranges $K$ and $K+\mathbf{h}$ respectively. To visualize $\hat{\rho}_x$ we compute its maximum values in the following way:
\begin{align*}\hat{\rho}_{x,max}(i)=\max_{1\leq k_1,k_2,k_3\leq i}(&\hat{\rho}_x(i,k_2,k_3),\hat{\rho}_x(k_1,i,k_3),\\
&\hat{\rho}_x(k_1,k_2,i)),i\geq 1.
\end{align*}
The values of $\hat{\rho}_{x,max}$(black color), $\hat{\rho}_{y,max}$(red color), $\hat{\rho}_{z,max}$(blue color) are given in Figure \ref{Fig:cov}. 
{We choose the value of $m$ in such a way that $\hat{\rho}_{x,max}(i)\leq \varepsilon_0,$ $\hat{\rho}_{y,max}(i)\leq \varepsilon_0,$ $\hat{\rho}_{z,max}(i)\leq \varepsilon_0,$ for all $i\geq m,$ where $\varepsilon_0$ is a threshold. Here we use   $\varepsilon_0=0.04$ and obtain $m=9.$ We have $M_0=0.5$ and assume that $\sigma^2=0.2.$}

Moreover, due to simulation experiments in Section 4, we compute critical values for the change-point statistics $T_W(\cdot)$ and $p-$values from inequality \eqref{critalv} with $m=7.$ 
\begin{figure}
    \centering
    \includegraphics[width=0.55\textwidth]{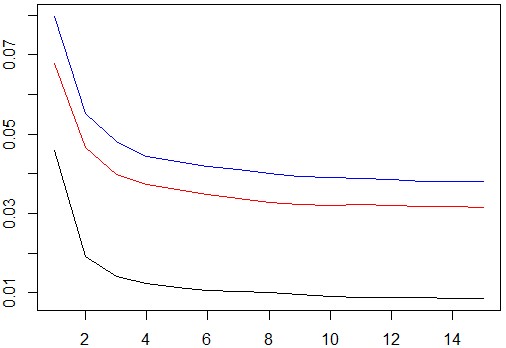}
    \caption{Values of empirical covariance $\hat{\rho}_{x,max},\hat{\rho}_{y,max},\hat{\rho}_{z,max}.$}
    \label{Fig:cov}
\end{figure}

For the random field of estimated local entropies we have $m=1$ and the parametric set $\Theta_0$ is constructed with $\Delta_0=\Delta_1=2,$ $\gamma_0=0.05,\gamma_1=0.5$ and $L_M= 4$ in \eqref{parset}, which gives $|\Theta_0|=12366.$ We assume that $\sigma^2=0.5$ and $M_0=\sigma.$ 
The result of our change point analysis is presented in Table \ref{tb4}. 
\begin{table}[h]
    \centering
    \begin{tabular}{c r r r }
        \hline\noalign{\smallskip}
        Attribute &  Sample var. & Test statistics & $p-$value \\
        \noalign{\smallskip}\hline\noalign{\smallskip}
        $\tilde{x}$&0.04589 & 0.15995 & 1.00\\
        $\tilde{y}$&0.06795 & 0.44733 &$2.1\times10^{-10}$\\
        $\tilde{z}$& 0.07982 &0.43383 &$1.3\times10^{-6}$\\
        $E$& 0.30126 &0.46811 &$3.96\times10^{-8}$\\
        \noalign{\smallskip}\hline
    \end{tabular}
    \caption{Change-point test for mean local directions of real data.}
    \label{tb4}
\end{table}
Our change point test detects the evidence of anomaly regions in real fibre data at significance level {$\alpha=8.4\times 10^{-10}.$}

 Similarly to the case of RSA data, the detection of anomalies by the $3\sigma-$rule  gives meaningless results. The Spatial SAEM algorithm works much better. Its results are presented in Figures \ref{fig:Picture25a}, \ref{fig:Picture25b}, and \ref{fig:Picture25c}, where the color labelling of points is the same as for the simulated data.
\begin{figure}[h]
			\center{\includegraphics[width=0.8\linewidth,keepaspectratio]{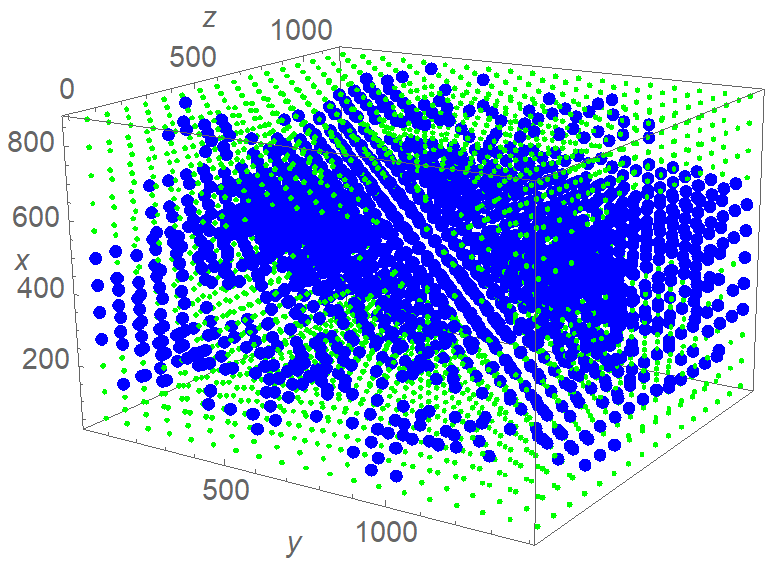} }
		\caption{Anomaly detection in the fibre image (Fig. \ref{fig:Picture22l}) using local entropy.}
		\label{fig:Picture25a}
\end{figure}
\begin{figure}[h]
	\center{\includegraphics[width=0.8\linewidth,keepaspectratio]{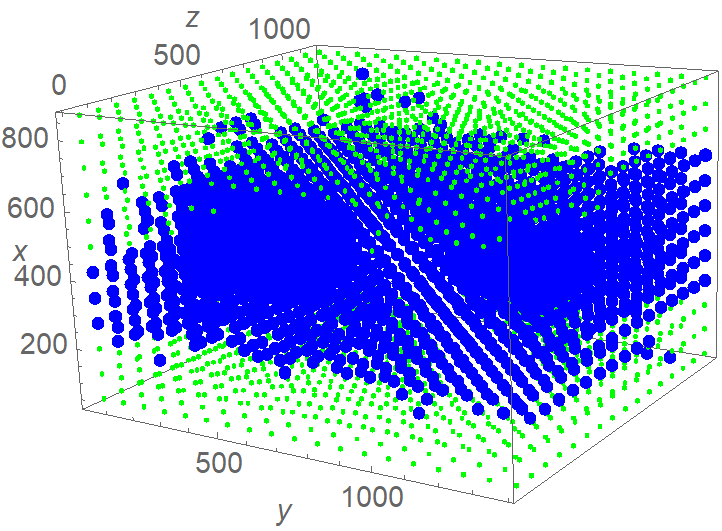}}
		\caption{Anomaly detection in the fibre image (Fig. \ref{fig:Picture22l}) using mean of local direction.}
		\label{fig:Picture25b}
\end{figure}	
\begin{figure}[h]
		\center{\includegraphics[width=0.8\linewidth,keepaspectratio]{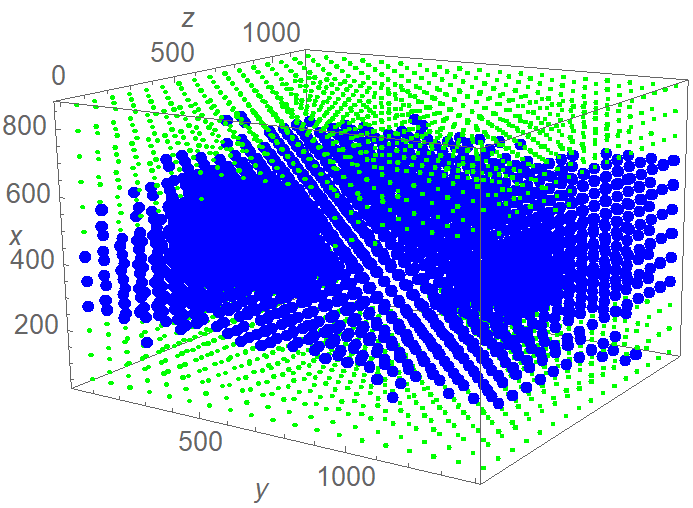}}
		\caption{Anomaly detection in the fibre image (Fig. \ref{fig:Picture22l}) using local entropy and mean of local direction.}
	\label{fig:Picture25c}
\end{figure}    

We conclude that the Spatial SAEM algorithm with  attributes entropy, MLD, and a combination of both produces adequate results.

Since the images in Figure \ref{fig:Picture25b} and \ref{fig:Picture25c}  look very similar at first glance, we investigate the results of the Spatial SAEM anomaly detection in more detail. We separate a part of the 3D image (Figure \ref{fig:Picture22r}) into 9 layers and present the result of clustering for the 1st and 5th layer for the entropy in Figure~\ref{fig:Picture27}, for MLD in Figure~\ref{fig:Picture28}, and for the combination of both in Figure~\ref{fig:Picture29}.

\begin{figure*}[h]
\begin{minipage}[h]{0.48\linewidth}
		\includegraphics[width=\textwidth,keepaspectratio]{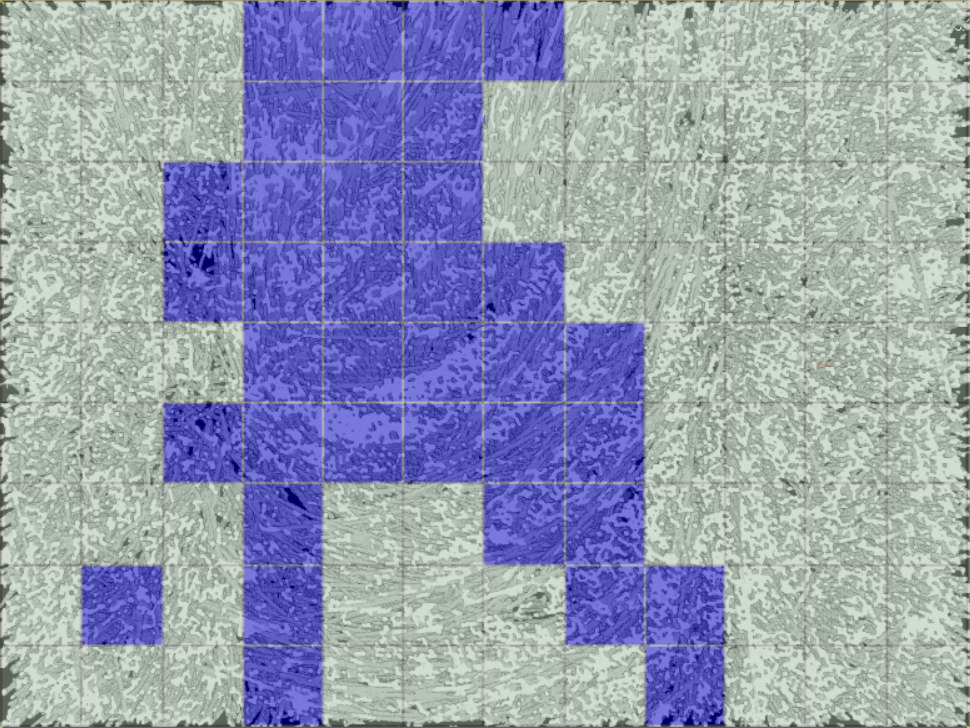} \center{1st layer}
	\end{minipage}
	\hfill
	\begin{minipage}[h]{0.48\linewidth}
		\includegraphics[width=\textwidth,keepaspectratio]{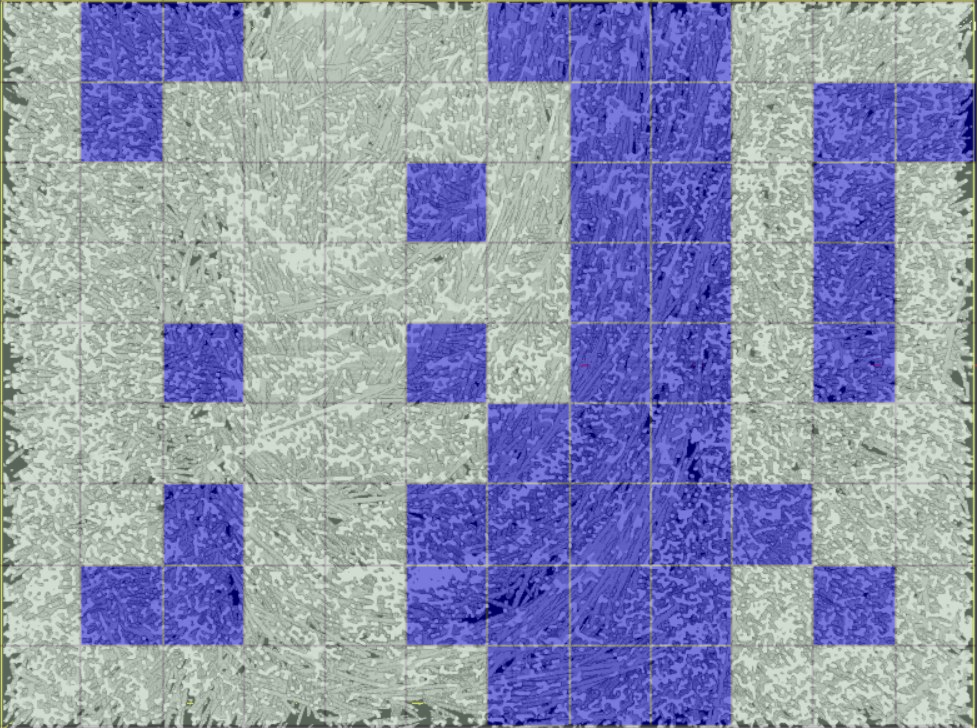} \center{5th layer}
	\end{minipage}
	\caption{Spatial SAEM clustering according to entropy}
	\label{fig:Picture27}
\end{figure*}
\begin{figure*}[h]
\begin{minipage}[h]{0.48\linewidth}
		\includegraphics[width=\textwidth,keepaspectratio]{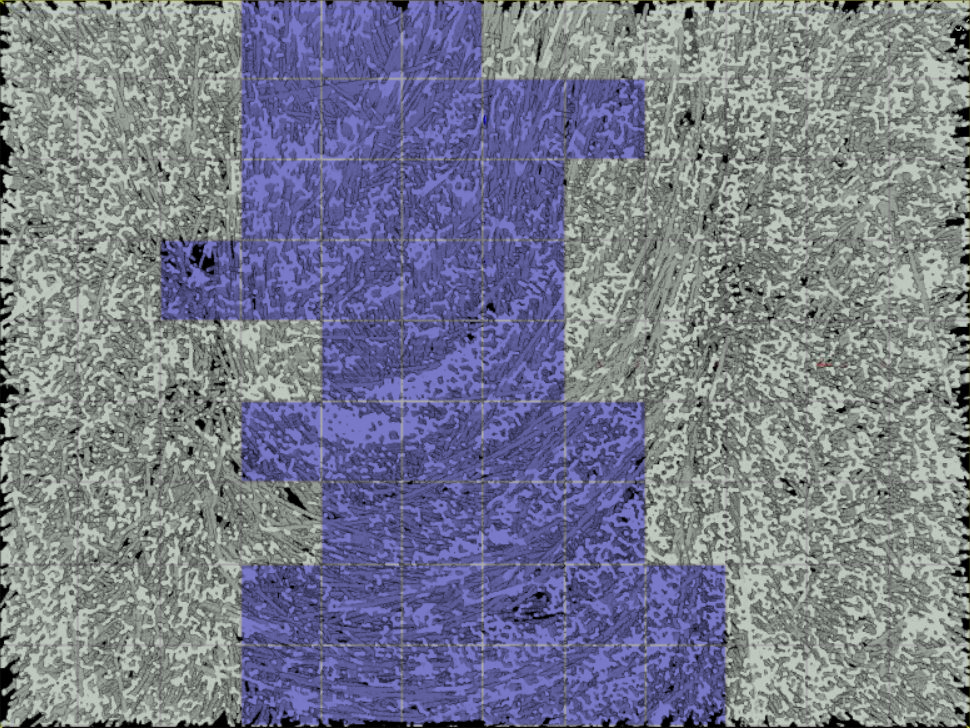} \center{1st layer}
	\end{minipage}
	\hfill
	\begin{minipage}[h]{0.48\linewidth}
		\includegraphics[width=\textwidth,keepaspectratio]{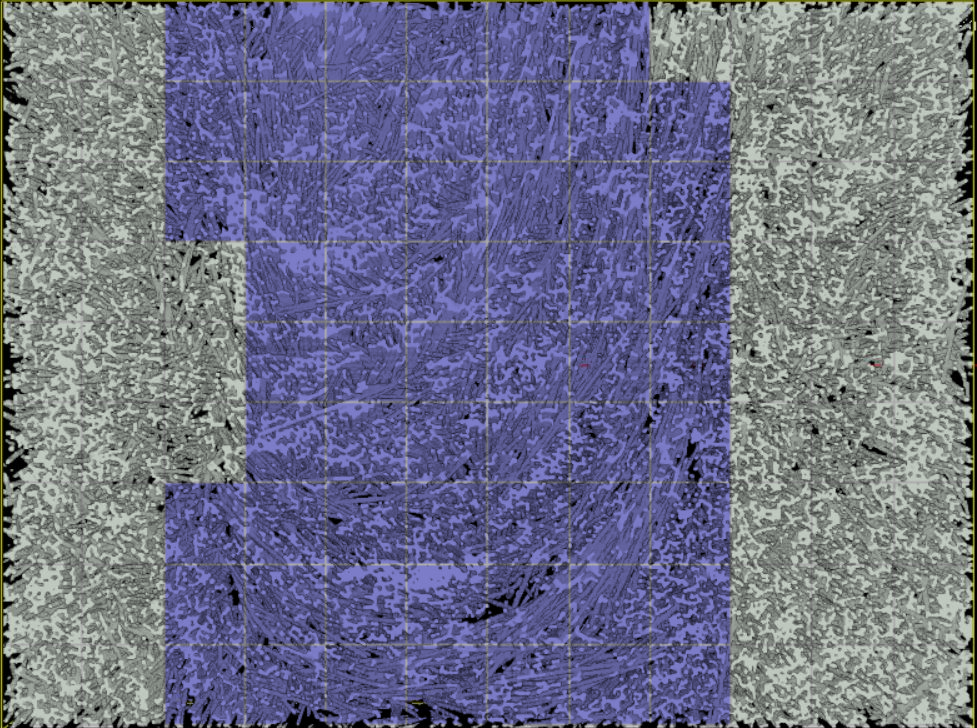} \center{5th layer}
	\end{minipage}
	\caption{Spatial SAEM clustering according to mean of local directions}
	\label{fig:Picture28}
\end{figure*}
\begin{figure*}[h]
\begin{minipage}[h]{0.48\linewidth}
		\includegraphics[width=\textwidth,keepaspectratio]{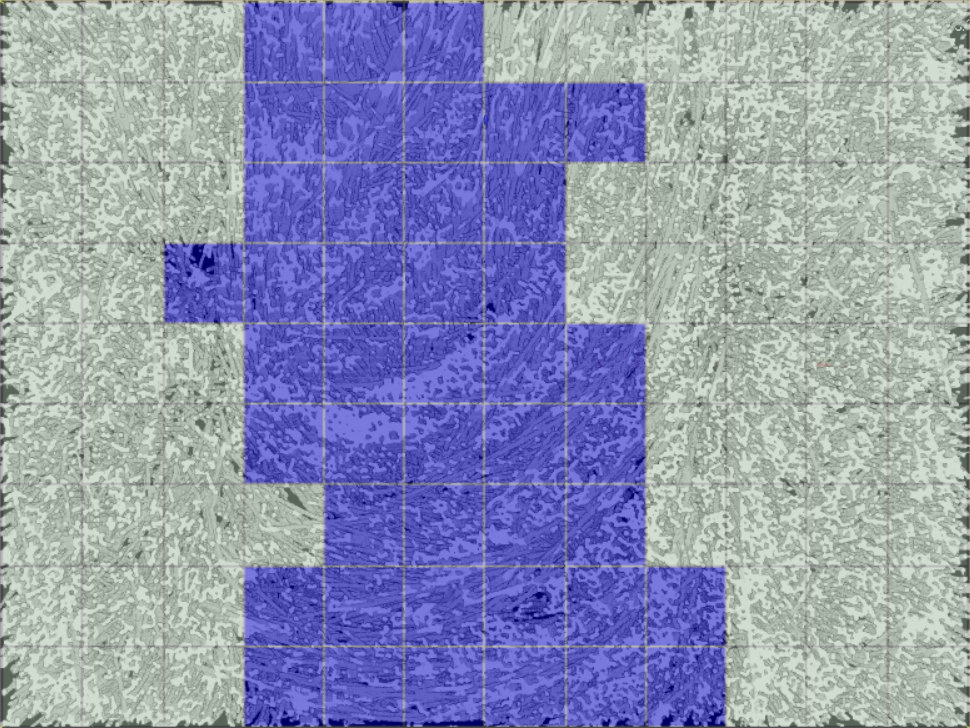} \center{1st layer}
	\end{minipage}
	\hfill
	\begin{minipage}[h]{0.48\linewidth}
		\includegraphics[width=\textwidth,keepaspectratio]{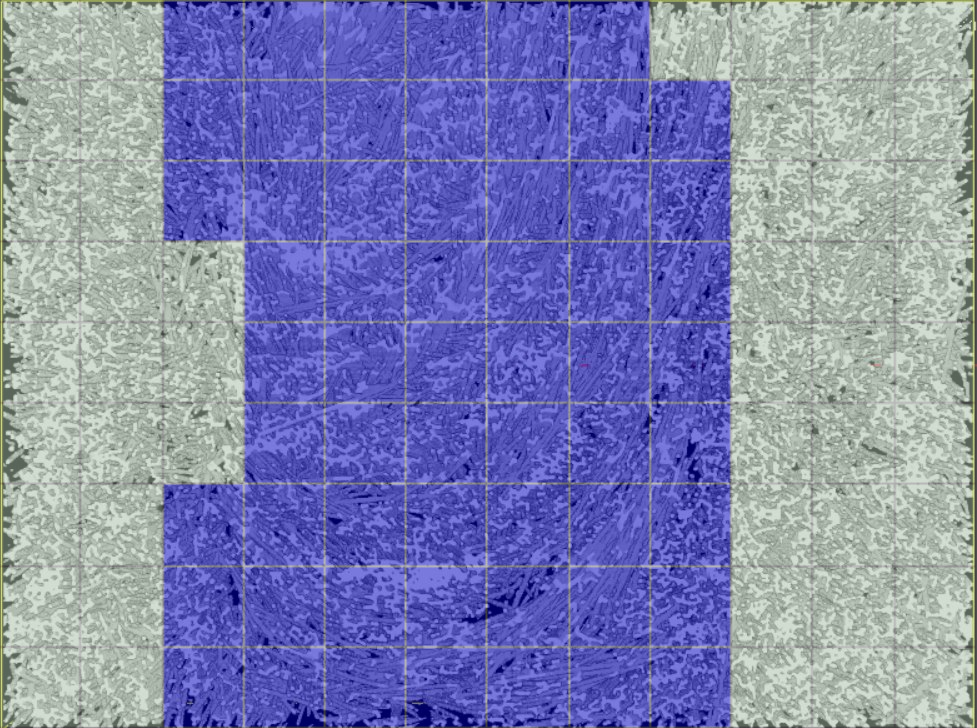} \center{5th layer}
	\end{minipage}
	\caption{Spatial SAEM clustering according to a combination of entropy and MLD}
	\label{fig:Picture29}
\end{figure*}
One can observe that the Spatial SAEM algorithm with local entropy attribute detects vortices of fibers in the material as anomaly regions. This is natural since a vortex exhibits a large diversity of fibre directions, and the entropy is a measure of such diversity. The Spatial SAEM anomaly detection using the mean of local fibre directions identifies the central part of the image (cf. Figure \ref{fig:Picture28}) as an anomaly region, where the directions of fibres differ from the average throughout the image. Finally, the Spatial SAEM approach using both clustering attributes identifies both vortices of fibres and layers of fibres with principally different main direction, cf. Figure \ref{fig:Picture29}.

\begin{acknowledgements}
We are grateful to Dr. Stefanie Schwaar from Fraunhofer ITWM for valuable discussions and to Jan Niedermeyer for simulating RSA data.
\end{acknowledgements}

\bibliographystyle{abbrv}
\bibliography{Literatur,bib_changepoint,bib_changepoint_RF}


%
%

\end{document}